\documentclass[11pt,journal,twoside,onecolumn,draftcls]{IEEEtran}
\usepackage{amsfonts,amsmath,amssymb, mathrsfs}
\usepackage{epsf}
\usepackage{epsfig}
\usepackage{graphics}

\newtheorem{theorem}{Theorem}

\newtheorem{observation}{Observation}

\newtheorem{corollary}{Corollary}
\newtheorem{remark}{Remark}

\date{}

\usepackage{array}

\begin{document}
\title{Generating Dependent Random Variables Over Networks}
\author{
\authorblockN{Amin Aminzadeh Gohari, Venkat Anantharam\\}
\authorblockA{aminzadeh@sharif.edu, ananth@eecs.berkeley.edu
}}

\maketitle
\begin{abstract}  In this paper we study the problem of generation of dependent random variables, known as the ``coordination capacity" \cite{CoverPermuter} \cite{Cuff}, in multiterminal networks. In this model $m$ nodes of the network are observing i.i.d. repetitions of $X^{(1)}$, $X^{(2)}$,..., $X^{(m)}$ distributed according to
$q(x^{(1)},...,x^{(m)})$. Given a joint distribution
$q(x^{(1)},...,x^{(m)},y^{(1)},...,y^{(m)})$, the final goal of the $i^{th}$ node is to construct the i.i.d. copies of $Y^{(i)}$ after the communication over the network where $X^{(1)}$,
$X^{(2)}$,..., $X^{(m)}, Y^{(1)}$, $Y^{(2)}$,..., $Y^{(m)}$ are
jointly distributed according to $q(x^{(1)},...,x^{(m)},y^{(1)},...,y^{(m)})$. To do this, the nodes can exchange messages over the network at rates not exceeding the capacity constraints of the links. This problem is difficult to solve even for the special case of two nodes. In this paper we prove new inner and outer bounds on the achievable rates for networks with two nodes.
\end{abstract}

\section{Introduction}
Traditionally coding is used in networks to distribute information among
the nodes. However in some applications (certain) nodes of the network
may need to coordinate with other nodes
to carry out some joint action
\cite{AB}, \cite{CoverPermuter}.
To accomplish this coordination, the nodes could in general talk back and forth
with each other. Given capacity constraints on the links between
the nodes, the goal may be to minimize the amount of communication
needed to accomplish the coordination.
This is not a traditional communication problem.
There is no explicit message to be transmitted,
and any node or set of nodes involved in the coordination
does not necessarily have to figure out the task
to be performed by the other nodes.
It is valuable to develop a general framework that
includes both the traditional demands of transmission of
(possibly correlated) sources over networks, and coordination demands.
Such a model has started to evolve in the recent literature
\cite{CoverPermuter, Cuff}, see also \cite{KramerSavari}.
Our formulation can be understood as a natural extension of these works.

Roughly speaking, we assume that the $m$ nodes of the network are observing i.i.d. repetitions of $X^{(1)}$,
$X^{(2)}$,..., $X^{(m)}$ distributed according to
$q(x^{(1)},...,x^{(m)})$. Given a joint distribution
$q(x^{(1)},...,x^{(m)},y^{(1)},...,y^{(m)})$, the final goal of the $i^{th}$ node is to compute the i.i.d. copies of $Y^{(i)}$ after the communication over the network where $X^{(1)}$,
$X^{(2)}$,..., $X^{(m)}, Y^{(1)}$, $Y^{(2)}$,..., $Y^{(m)}$ are
jointly distributed according to
$\\q(x^{(1)},...,x^{(m)},y^{(1)},...,y^{(m)})$. To do this, the nodes can exchange messages over the network at rates not exceeding the capacity constraints of the links. We assume that there is a directed edge from the $i^{th}$ node to the $j^{th}$ node with the rate constraint $R_{i,j}$ for $1\leq i,j \leq m$.
We say that the network with the set of
rates $R_{i,j}$ for $1\leq i,j \leq m$ is admissible if the $i^{th}$
node is able to create the i.i.d. copies of $Y^{(i)}$ within a
vanishing total variation distance. All (or a subset) of the nodes of the network may have access to \emph{common randomness} at a certain rate in a given application. This depends on the possibility of a parallel resource that provides common randomness to the nodes at some constant rate. For simplicity we assume that \emph{no common randomness is
provided to the nodes}, although this assumption is not in principle necessary and can be further explored. On the other hand, private randomization at the individual nodes can very well be feasible in a practical network. Thus, we provide the nodes with such ability throughout this paper.

Note that the traditional problem of communication of messages over a network
can be thought of as an special case of the problem defined above.
One just needs to suitably choose $q(x^{(1)},...,x^{(m)},y^{(1)},...,y^{(m)})$
where the i.i.d. copies $X^{(i)}$ represent the information initially available to
the $i^{th}$ node, and the i.i.d. copies $Y^{(i)}$ represent the
desired information of the $i^{th}$ node after the communication.
Therefore the advantage of network coding over pure routing, or the
insufficiency of linear codes are of relevance here as well.
Lastly we note that there are some similarities between the problem
of ``communication complexity" and this problem when $Y^{(i)}$s are functions of $(X^{(1)},...,X^{(m)})$.

The simplest model to consider is the network with two nodes.
This special case is also interesting because one can
partition the nodes of a large network into two sets, $S$ and $S^c$, and
treat the nodes on each side of the partition as a single supernode.
Thus, any outer bound on the two-node problem results in an outer bound for arbitrary networks. Still, the problem is difficult to solve even in the special case of two nodes. In this paper, we restrict ourselves to the special case of two nodes where we prove new lower and upper bounds. Our results build upon and generalize some of the results in \cite{CoverPermuter} and \cite{Cuff}.

This paper is organized as follows. In section \ref{sec:definition}, we introduce the basic notation and
definitions used in this paper. Section \ref{Section:MainResults}
contains the main results of the paper, and section
\ref{Section:Proofs} includes the proofs.

\section{Definitions}\label{sec:definition}
In this section we provide a rigorous definition for the model described in the
introduction for the case of two nodes. A general model for the case of $m$
terminals can be defined along the same lines, as sketched in the preceding
section. This we omit, because we prove results only for the case of two nodes.

Take an arbitrary $q(x^{(1)}, x^{(2)}, y^{(1)}, y^{(2)})$. Given an arbitrary natural number $r$, we say
that $(R_{12}, R_{21})$ is admissible with \emph{$r$ interactive
rounds of communication} if for any positive $\epsilon$, there is a
natural number $n_0$ where the following holds for any $n\geq n_0$: the two nodes
observe $n$ i.i.d. copies of $X^{(1)}, X^{(2)}$, i.e.
$X^{(1)}_{1:n}, X^{(2)}_{1:n}$. Assume that the first node is using
the external randomness $M_1$, and the second node is using the
external randomness $M_2$. Random variables $M_1$, $M_2$, and
$X^{(1)}_{1:n}X^{(2)}_{1:n}$ must be mutually independent. Let $C_1,
C_2, ..., C_r$ denote the interactive communication used in the
code, and let $\textbf{C}=(C_1,
C_2, ..., C_r)$. We have $H(C_i|C_{1:i-1}X^{(1)}_{1:n}M_1)=0$ if $i$ is odd,
$H(C_i|C_{1:i-1}X^{(2)}_{1:n}M_2)=0$ if $i$ is even. We further have
$$\frac{1}{n}\sum_{i: odd}H(C_i)\leq R_{12}$$
and
$$\frac{1}{n}\sum_{i: even}H(C_i)\leq R_{21}.$$
Following the communication, the first node creates
$\widehat{Y}^{(1)}_{1:n}$, and the second node creates
$\widehat{Y}^{(2)}_{1:n}$ such that the total variation distance
between $p(X^{(1)}_{1:n}, X^{(2)}_{1:n}, \widehat{Y}^{(1)}_{1:n},
\widehat{Y}^{(2)}_{1:n})$ and the distribution of $(X_{1:n}^{(1)},
X_{1:n}^{(2)}, Y_{1:n}^{(1)}, Y_{1:n}^{(2)})$, constructed by taking
$n$ i.i.d copies of $q(x^{(1)}, x^{(2)}, y^{(1)}, y^{(2)})$, is less
than or equal to $\epsilon$. In other words
\begin{align*}&\sum_{x^{(1)}_{1:n}, x^{(2)}_{1:n}, y^{(1)}_{1:n}, y^{(2)}_{1:n}}\\&\big|p(X^{(1)}_{1:n}=x^{(1)}_{1:n}, X^{(2)}_{1:n}=x^{(2)}_{1:n}, \widehat{Y}^{(1)}_{1:n}=y^{(1)}_{1:n},
\widehat{Y}^{(2)}_{1:n}=y^{(2)}_{1:n})-\\&p(X^{(1)}_{1:n}=x^{(1)}_{1:n}, X^{(2)}_{1:n}=x^{(2)}_{1:n}, Y^{(1)}_{1:n}=y^{(1)}_{1:n},
Y^{(2)}_{1:n}=y^{(2)}_{1:n})\big|\leq \epsilon.
\end{align*}

If $r=1$, we have only one communication from the first node to the
second node.

For the special case of $Y^{(1)}=X^{(1)}$ and
$X^{(2)}=constant$ and $r=1$, one gets a problem that is a special
case of the problem considered by Paul Cuff \cite{Cuff}.

\section{Statement of the Results}\label{Section:MainResults}

Consider the case of $r$ rounds of interactive
communication, meaning that the two nodes talk back and forth for
$r$ rounds.
\subsection{Converse results}
$\bullet$ \emph{The case of $r=1$ (i.e. $R_{21}=0$):}

When $Y^{(1)}=X^{(1)}$, $X^{(2)}=constant$, one gets a problem that
is a special case of Paul Cuff's result \cite{Cuff}. Cuff showed
that in this case the minimum admissible $R_{12}$ is equal to the
Wyner common information \cite{Wyner}, i.e.
\begin{eqnarray*}\inf_{Y^{(2)}- U- X^{(1)}}I(U;X^{(1)}Y^{(2)}).\end{eqnarray*}
We generalize this result for arbitrary $q(x^{(1)}, x^{(2)},
y^{(1)}, y^{(2)})$:
\begin{theorem}\label{Thm:1} Assume that $r=1$, that is when $R_{21}=0$. Then one must have $Y^{(1)}- X^{(1)}- X^{(2)}$, and furthermore $(R_{12},0)$
belongs to $\mathcal{R}_{1}$ defined as follows:\footnote{The authors would like to thank Mohammad Hossein Yassaee for pointing out a typo in the statement of this theorem.}
\begin{eqnarray*}\big\{(R_{12},0)\in \mathcal{R}^2&:& \exists
p(u, x^{(1)}, x^{(2)}, y^{(1)}, y^{(2)})\in T_1 \mbox{ s.t.
}\\R_{12}&\geq& I(U;X^{(1)}Y^{(1)}Y^{(2)}|X^{(2)})\big\},
\end{eqnarray*}
where
\begin{eqnarray*}
&T_1\triangleq\big\{&p(u, x^{(1)}, x^{(2)}, y^{(1)}, y^{(2)}):\\&&
X^{(1)}, X^{(2)}, Y^{(1)}, Y^{(2)}\sim q(x^{(1)}, x^{(2)}, y^{(1)},
y^{(2)}),\\&& U- X^{(1)}- X^{(2)},\\&&
Y^{(1)}- UX^{(1)}- X^{(2)}Y^{(2)},\\&&
Y^{(2)}- UX^{(2)}- X^{(1)}Y^{(1)},\\&&|\mathcal{U}|\leq
|\mathcal{X}^{(1)}||\mathcal{Y}^{(1)}||\mathcal{X}^{(2)}||\mathcal{Y}^{(2)}|+1\big\}.
\end{eqnarray*}
\end{theorem}
$\bullet$ \emph{The case of a fixed $r\geq 2$:}
\begin{theorem}\label{Thm:2} Take an arbitrary $q(x^{(1)}, x^{(2)},
y^{(1)}, y^{(2)})$. Then any admissible pair $(R_{12}, R_{21})$ must
belong to $\mathcal{R}_{2}(r)$ defined as the convex closure of rate pairs $(R_{12}, R_{21})\in \mathcal{R}^2$ such that there exists
$p(f_1, ..., f_r, x^{(1)}, x^{(2)}, y^{(1)}, y^{(2)})\in T_2(r)
$ such that \begin{eqnarray*}R_{12}&\geq& I(X^{(1)};
\textbf{F}|X^{(2)}),\\R_{21}&\geq&I(X^{(2)};
\textbf{F}|X^{(1)}),\\
R_{12}+R_{21}&\geq& I(X^{(1)};
\textbf{F}|X^{(2)})+\\&&I(X^{(2)}; \textbf{F}|
X^{(1)})+I(Y^{(1)};Y^{(2)}|X^{(1)}X^{(2)}),
\end{eqnarray*}
where $\textbf{F}=(F_1,
F_2, ..., F_r)$ and $T_2(r)$ is the set of $p(f_1, f_2, ..., f_r, x^{(1)}, x^{(2)},
y^{(1)}, y^{(2)})$ satisfying
\begin{eqnarray*}&
X^{(1)}, X^{(2)}, Y^{(1)}, Y^{(2)}\sim
q(x^{(1)}, x^{(2)}, y^{(1)}, y^{(2)}),\\&
F_i- F_{1:i-1}X^{(1)}- X^{(2)} \mbox{ if } i \mbox{ is odd},\\&
F_i- F_{1:i-1}X^{(2)}- X^{(1)} \mbox{ if } i \mbox{
is even},\\& Y^{(1)}- \textbf{F}X^{(1)}- X^{(2)}Y^{(2)},\\
&Y^{(2)}- \textbf{F}X^{(2)}- X^{(1)}Y^{(1)},\\&\forall\ 1\leq
i\leq r:\\&\ \ \ |\mathcal{F}_i|\leq
|\mathcal{F}_{i-1}||\mathcal{F}_{i-2}|\cdot\cdot\cdot
|\mathcal{F}_{1}||\mathcal{X}^{(1)}||\mathcal{Y}^{(1)}||\mathcal{X}^{(2)}||\mathcal{Y}^{(2)}|+1.
\end{eqnarray*}
\end{theorem}
$\bullet$ \emph{The case of no constraints on the
value of $r$:}
\begin{theorem}\label{Thm:3} Take an arbitrary $q(x^{(1)}, x^{(2)},
y^{(1)}, y^{(2)})$. Then for any $r$, any admissible pair $(R_{12},
R_{21})$ must belong to $\mathcal{R}_{3}$ defined as the convex
closure of
\begin{eqnarray*}\big\{(R_{12}, R_{21})\in \mathcal{R}^2&:& \exists
p(u, x^{(1)}, x^{(2)}, y^{(1)}, y^{(2)})\in T_3 \mbox{ s.t.
}\\R_{12}&\geq& I(X^{(1)}; U|X^{(2)}),\\R_{21}&\geq&I(X^{(2)};
U|X^{(1)}),\\
R_{12}+R_{21}&\geq& I(X^{(1)}; U|X^{(2)})+I(X^{(2)};U|
X^{(1)})\\&&+I(Y^{(1)};Y^{(2)}|X^{(1)}X^{(2)})\big\},
\end{eqnarray*}
where $T_3$ is the set of $p(u, x^{(1)}, x^{(2)}, y^{(1)}, y^{(2)})$ satisfying
\begin{eqnarray*}
\\&
X^{(1)}, X^{(2)}, Y^{(1)}, Y^{(2)}\sim q(x^{(1)}, x^{(2)}, y^{(1)},
y^{(2)}),\\& I(X^{(1)}; X^{(2)}|U)\leq I(X^{(1)};
X^{(2)}),\\& Y^{(1)}- UX^{(1)}- X^{(2)}Y^{(2)},\\
&Y^{(2)}- UX^{(2)}- X^{(1)}Y^{(1)},\\&|\mathcal{U}|\leq
|\mathcal{X}^{(1)}||\mathcal{Y}^{(1)}||\mathcal{X}^{(2)}||\mathcal{Y}^{(2)}|+1.
\end{eqnarray*}
\end{theorem}

\section{Achievability results}
\begin{theorem}\label{Thm:Achievability1} Take an arbitrary $q(x^{(1)}, x^{(2)},
y^{(1)}, y^{(2)})$ such that $X^{(1)}=X^{(2)}$, and let $X=X^{(1)}=X^{(2)}$. Assume that random variables $F_1, F_2, ..., F_r$ are jointly distributed with $X,
Y^{(1)}, Y^{(2)}$ such that \begin{eqnarray*}&Y^{(1)}- F_rF_{r-1}...F_1X- Y^{(2)}.\end{eqnarray*} Furthermore assume that positive reals $R'_1$, $R'_2$, ...,$R'_r$ satisfy the following:
\begin{eqnarray*}\forall 1\leq s\leq r:~&\sum_{i=1}^s R'_i > I(Y^{(1)}, Y^{(2)};F_1, ..., F_s|X)
\end{eqnarray*}
Then the rate pair $(\sum_{odd~i}R'_i, \sum_{even~i}R'_i)$ is achievable with two rounds of communication, i.e. with $r=2$.
\end{theorem}
\emph{Discussion:} The above result is based on a generalization of Lemma 6.1 of \cite{Cuff}, which itself is a consequence of the resolvability
work of Han and Verdu in \cite{HanVerdu} and Wyner \cite{Wyner}.

\emph{Lemma 6.1 of \cite{Cuff}}: For any discrete distribution $p(f, v)$ and
each $n$, let $\mathcal{C}(n) = {F_{1:n}(t)}_{t=1}^{2^{nR}}$ be a ``codebook" of sequences each independently drawn according to $\prod_{k=1}^n
p_F(f_k)$. For a fixed codebook, define the distribution
$$Q(v_{1:n}) = 2^{-nR}\sum_{t=1}^{2^{nR}}\prod_{k=1}^n p_{V|F}(v_k|f_k(t)).$$
Then if $R > I(V ;F)$, $$\lim_{n-\infty} \mathbb{E}\bigg\|Q(v_{1:n}) -\prod_{k=1}^np_V(v_k)\bigg\|_1= 0,$$
where the expectation is with respect to the randomly
constructed codebooks $\mathcal{C}(n)$.

We now present the generalization of the above lemma that we use. To get back the above lemma, simply set $r=1$ and $X$ to be constant. Note that Theorem \ref{Thm:Achievability1} can be proved by taking $V$ to be $Y^{(1)}, Y^{(2)}$ in the following lemma.

\emph{Generalization of the lemma}: Take some arbitrary discrete joint distribution $p(f_1,f_2,f_3,...,f_r, v,x)$. For each $n$ and every sequence $x_{1:n}\in \mathcal{X}^n$, we randomly construct a ``codebook" $\mathcal{C}(x_{1:n})$ of size $2^{nR'_1+nR'_2+...+nR'_r}$ as follows: generate $2^{nR'_1}$ sequences of length $n$ independently, with each sequence distributed according to $\prod_{i=1}^nq(f_{1i}|x_i)$. Let us label these sequences by $f_{1,1:n}(t_1)$ for $1\leq t_1\leq 2^{nR'_1}$. The sequences $f_{1,1:n}(t_1)$ (for $1\leq t_1\leq 2^{nR'_1}$) will be the cloud centers for the sequences generated in the next step. For every $1\leq t_1\leq 2^{nR'_1}$, generate $2^{nR'_2}$ sequences of length $n$ using the conditional distribution of $\prod_{i=1}^nq(f_{2i}|f_{1i}(t_1),x_i)$. Let us label these sequences by $f_{2,1:n}(t_1,t_2)$ for $1\leq t_1\leq 2^{nR'_1}, 1\leq t_2\leq 2^{nR'_2}$. The sequences $f_{2,1:n}(t_1,t_2)$ (for $1\leq t_1\leq 2^{nR'_1}, 1\leq t_2\leq 2^{nR'_2}$) will be the cloud centers for the sequences generated in the next step. We continue this for $r$ steps. In the $r$-th step we generate $f_{r,1:n}(t_1,t_2,t_3,...,t_r)$ for $1\leq t_1\leq 2^{nR'_1}, 1\leq t_2\leq 2^{nR'_2}, 1\leq t_3\leq 2^{nR'_3}, ..., 1\leq t_r\leq 2^{nR'_r}$.

For a fixed codebook, define the distribution
$$Q(v_{1:n}|x_{1:n}) = 2^{-nR'_1-nR'_2-...-nR'_r}\sum_{t_1=1}^{2^{nR'_1}}\sum_{t_2=1}^{2^{nR'_2}}...\sum_{t_r=1}^{2^{nR'_r}}$$$$\prod_{k=1}^np_{V|F_1,...,F_r,X}(v_k|f_{1k}(t_1), f_{2k}(t_1, t_2),..., f_{rk}(t_1,t_2,...,t_r),x_k).$$
Then if \begin{eqnarray*}\forall 1\leq s\leq r:~&\sum_{i=1}^s R'_i > I(V;F_1, ..., F_s|X),
\end{eqnarray*} we have $$\lim_{n-\infty} \mathbb{E}\bigg\|Q(v_{1:n},x_{1:n}) -\prod_{k=1}^np_{V,X}(v_k,x_k)\bigg\|_1= 0,$$
where the expectation is with respect to the randomly
constructed codebooks $\mathcal{C}(x_{1:n})$ and i.i.d.\ sequences $X_{1:n}\sim \prod_{k=1}^np(x_k)$.

\begin{remark} For the case of $X=X^{(1)}=X^{(2)}$ and $r=1$ the above theorem implies that the rate $(R_{12},0)$ is achievable if $R_{12}>\inf_{U: Y^{(1)}- UX- Y^{(2)}} I(Y^{(1)}, Y^{(2)};U|X)$. Compare this with the converse given in Theorem \ref{Thm:1} which shows that any achievable $(R_{12},0)$ must satisfy $R_{12}\geq \inf_{U: Y^{(1)}- UX- Y^{(2)}} I(Y^{(1)}, Y^{(2)};U|X)$.
\end{remark}
\begin{remark}\label{remark1} The above theorem can be extended to scenarios where there are functions $k_1$ and $k_2$ satisfying $K=k_1(X^{(1)})=k_2(X^{(2)})$ and $X^{(1)}X^{(2)}- K- Y^{(1)}Y^{(2)}$. In this case the nodes can create $K$ and work with it to reconstruct i.i.d copies of $Y^{(1)}$ and $Y^{(2)}$ within an arbitrarily small total variational distance from the i.i.d. scenario. The constructed outputs together with $X^{(1)}, X^{(2)}$ will also have small total variational distance from the i.i.d scenario.
\end{remark}

\begin{theorem}\label{Thm:Achievability2} Take some $Z^{(1)}, Z^{(2)}$ arbitrarily distributed with $X^{(1)},X^{(2)},Y^{(1)},Y^{(2)}$. Assume that the rate pair $(R_{12}, R_{21})$ is achievable for the problem of $q(x^{(1)},x^{(2)},z^{(1)},z^{(2)})$ with $r_1$ rounds of communication, and the rate pair $(R'_{12}, R'_{21})$ is achievable for the problem of $q(x^{(1)}z^{(1)},x^{(2)}z^{(2)},y^{(1)},y^{(2)})$ with $r_2$ rounds of communication. Then the rate pair $(R_{12}+R'_{12}, R_{21}+R'_{21})$ is achievable for the problem of $q(x^{(1)},x^{(2)},y^{(1)},y^{(2)})$ with $r_1+r_2$ rounds of communication.
\end{theorem}
\begin{observation} If $H(Y^{(1)}|X^{(2)})=0$ and $H(Y^{(2)}|X^{(1)})=0$, then the set of $(R_{12}, R_{21})$ satisfying $R_{21}> H(Y^{(1)}|X^{(1)})$, $R_{12}> H(Y^{(2)}|X^{(2)})$ is achievable using Slepian-Wolf binning with two rounds of communication, i.e. with $r=2$. Small probability of block recovery error implies a small total variational distance.
\end{observation}

\begin{corollary}\label{corrolary1} The above observation together with Theorems \ref{Thm:Achievability1} and \ref{Thm:Achievability2} imply an achievable rate region for the general problem of $q(x^{(1)},x^{(2)},y^{(1)},y^{(2)})$. The two nodes can first exchange $X^{(1)}$ and $X^{(2)}$, by using up rates $H(X^{(2)}|X^{(1)})$ and $H(X^{(1)}|X^{(2)})$. One can then apply Theorem \ref{Thm:Achievability1} for the reconstruction of $Y^{(1)}$ and $Y^{(2)}$.
\end{corollary}

\section{Proofs}\label{Section:Proofs}
\begin{proof}[Proof of Theorem \ref{Thm:1}]
Take an arbitrary $q(x^{(1)}, x^{(2)}, y^{(1)}, y^{(2)})$. Since
$(R_{12}, 0)$ is admissible with one-way communication from the
first node to the second node, for any positive $\epsilon$, a code
of length $n$ exists such that the communication rates are less than
or equal to $(R_{12}, 0)$ and the total variation distance between the overall
reconstruction and the distribution of $n$ i.i.d copies of
$(X^{(1)}, X^{(2)}, Y^{(1)}, Y^{(2)})$ is less than or equal to
$\epsilon$. In other words for some natural number $n$, the two
nodes observe $n$ i.i.d. copies of $X^{(1)}, X^{(2)}$, i.e.
$X^{(1)}_{1:n}, X^{(2)}_{1:n}$ respectively. Assume that the first node is using
the external randomness $M_1$, and the second node is using the
external randomness $M_2$. $M_1$, $M_2$ and
$X^{(1)}_{1:n}X^{(2)}_{1:n}$ must be mutually independent. Let $C$
denote the message sent from the first node to the second node. We
have $H(C|X^{(1)}_{1:n}M_1)=0$ and
$$\frac{1}{n}H(C)\leq R_{12}.$$
Following the communication, the first node creates
$\widehat{Y}^{(1)}_{1:n}$, and the second node creates
$\widehat{Y}^{(2)}_{1:n}$ such that the total variation distance
between $p(X^{(1)}_{1:n}, X^{(2)}_{1:n}, \widehat{Y}^{(1)}_{1:n},
\widehat{Y}^{(2)}_{1:n})$ and the distribution of $(X_{1:n}^{(1)},
X_{1:n}^{(2)}, Y_{1:n}^{(1)}, Y_{1:n}^{(2)})$, constructed by taking
$n$ i.i.d copies of $q(x^{(1)}, x^{(2)}, y^{(1)}, y^{(2)})$, is less
than or equal to $\epsilon$.

The proof generalizes the one given by Paul Cuff in \cite{Cuff}.
Take a random variable $J$ uniform on $\{1,2,3,...,n\}$ and
independent of all the above mentioned random variables. Let
\begin{eqnarray*}&U=JCX^{(2)}_{1:J-1}X^{(2)}_{J+1:n},\\&
X^{(1)}=X^{(1)}_J, X^{(2)}=X^{(2)}_J,
\widehat{Y}^{(1)}=\widehat{Y}^{(1)}_J,
\widehat{Y}^{(2)}=\widehat{Y}^{(2)}_J\end{eqnarray*} Note that the
total variation distance between the joint distribution of
$(X^{(1)}, X^{(2)}, \widehat{Y}^{(1)}, \widehat{Y}^{(2)})$ and that
of $(X^{(1)}, X^{(2)}, Y^{(1)}, Y^{(2)})$ distributed according to
$q(x^{(1)}, x^{(2)}, y^{(1)}, y^{(2)})$ is less than or equal to
$\epsilon$. This is true since for any $1\leq j\leq n$, the total
variation distance between the joint distribution of $(X_j^{(1)},
X_j^{(2)}, \widehat{Y}_j^{(1)}, \widehat{Y}_j^{(2)})$ and the joint
distribution of $(X_j^{(1)}, X_j^{(2)}, Y_j^{(1)}, Y_j^{(2)})$ is
less than or equal to $\epsilon$.

We further have
$$\widehat{Y}^{(2)}-UX^{(2)}-X^{(1)}\widehat{Y}^{(1)}$$
since
\begin{eqnarray}&0=I(\widehat{Y}^{(2)}_{1:n};X^{(1)}_{1:n}\widehat{Y}^{(1)}_{1:n}M_1|CM_2X^{(2)}_{1:n})=\label{eqn:1a}\\&
I(\widehat{Y}^{(2)}_{1:n}M_2;X^{(1)}_{1:n}\widehat{Y}^{(1)}_{1:n}M_1|CX^{(2)}_{1:n})-I(M_2;X^{(1)}_{1:n}\widehat{Y}^{(1)}_{1:n}M_1|CX^{(2)}_{1:n})\geq
\nonumber\\&I(\widehat{Y}^{(2)}_{1:n}M_2;X^{(1)}_{1:n}\widehat{Y}^{(1)}_{1:n}M_1|CX^{(2)}_{1:n})-I(M_2;X^{(1)}_{1:n}\widehat{Y}^{(1)}_{1:n}M_1C|X^{(2)}_{1:n})=
\label{eqn:2a}\\&I(\widehat{Y}^{(2)}_{1:n}M_2;X^{(1)}_{1:n}\widehat{Y}^{(1)}_{1:n}M_1|CX^{(2)}_{1:n})-I(M_2;X^{(1)}_{1:n}M_1|X^{(2)}_{1:n})=
I(\widehat{Y}^{(2)}_{1:n}M_2;X^{(1)}_{1:n}\widehat{Y}^{(1)}_{1:n}M_1|CX^{(2)}_{1:n})\geq\nonumber\\&
I(\widehat{Y}^{(2)}_{j};X^{(1)}_{j}\widehat{Y}^{(1)}_{j}|CX^{(2)}_{1:j-1}X^{(2)}_{j+1:n}X^{(2)}_{j}),\
\ \ \ \ \ \ \ \ \forall\ 1\leq j\leq n,\nonumber
\end{eqnarray}
where equations \ref{eqn:1a} and \ref{eqn:2a} hold because
$H(\widehat{Y}^{(2)}_{1:n}|CM_2X^{(2)}_{1:n})=0$, and
$H(C\widehat{Y}^{(1)}_{1:n}|M_1X^{(1)}_{1:n})=0$. Next, note that
\begin{eqnarray*}&R_{12}\geq \frac{1}{n}H(C)\geq \frac{1}{n}I(C;X^{(1)}_{1:n}\widehat{Y}^{(1)}_{1:n}\widehat{Y}^{(2)}_{1:n}|X^{(2)}_{1:n})
=\frac{1}{n}H(X^{(1)}_{1:n}\widehat{Y}^{(1)}_{1:n}\widehat{Y}^{(2)}_{1:n}|X^{(2)}_{1:n})
-\frac{1}{n}H(X^{(1)}_{1:n}\widehat{Y}^{(1)}_{1:n}\widehat{Y}^{(2)}_{1:n}|X^{(2)}_{1:n}C)
\end{eqnarray*}
The first term $\frac{1}{n}H(X^{(1)}_{1:n}\widehat{Y}^{(1)}_{1:n}\widehat{Y}^{(2)}_{1:n}|X^{(2)}_{1:n})$ is equal to
$-\frac{1}{n}H(X^{(2)}_{1:n})+\frac{1}{n}H(X^{(1)}_{1:n}\widehat{Y}^{(1)}_{1:n}\widehat{Y}^{(2)}_{1:n}X^{(2)}_{1:n})$.
Using Csisz\'{a}r-K\"{o}rner inequality \cite[Lemma
2.7]{CsiszárKörner} on $\frac{1}{n}H(X^{(1)}_{1:n}\widehat{Y}^{(1)}_{1:n}\widehat{Y}^{(2)}_{1:n}X^{(2)}_{1:n})$, we can bound this term, i.e.
$\frac{1}{n}H(X^{(1)}_{1:n}\widehat{Y}^{(1)}_{1:n}\widehat{Y}^{(2)}_{1:n}|X^{(2)}_{1:n})$,
from below by
\begin{eqnarray*}&-\frac{1}{n}H(X^{(2)}_{1:n})+\frac{1}{n}H(X^{(1)}_{1:n}Y^{(1)}_{1:n}Y^{(2)}_{1:n}X^{(2)}_{1:n})+\frac{\epsilon}{n}\log(\frac{\epsilon}{(|\mathcal{X}^{(1)}||\mathcal{X}^{(2)}||\mathcal{\widehat{Y}}^{(1)}||\mathcal{\widehat{Y}}^{(2)}|)^n})
\\&=H(X^{(1)}Y^{(1)}Y^{(2)}|X^{(2)})+\frac{\epsilon}{n}\log(\frac{\epsilon}{(|\mathcal{X}^{(1)}||\mathcal{X}^{(2)}||\mathcal{\widehat{Y}}^{(1)}||\mathcal{\widehat{Y}}^{(2)}|)^n})\\&\geq
H(X^{(1)}Y^{(1)}Y^{(2)}|X^{(2)})+\epsilon\log(\frac{\epsilon}{|\mathcal{X}^{(1)}||\mathcal{X}^{(2)}||\mathcal{\widehat{Y}}^{(1)}||\mathcal{\widehat{Y}}^{(2)}|})\\&\geq
H(X^{(1)}\widehat{Y}^{(1)}\widehat{Y}^{(2)}|X^{(2)})+2\epsilon\log(\frac{\epsilon}{|\mathcal{X}^{(1)}||\mathcal{X}^{(2)}||\mathcal{\widehat{Y}}^{(1)}||\mathcal{\widehat{Y}}^{(2)}|})
=H(X^{(1)}\widehat{Y}^{(1)}\widehat{Y}^{(2)}|X^{(2)})-2\kappa(\epsilon)
\end{eqnarray*}
where
$\kappa(\epsilon)=-\epsilon\log(\frac{\epsilon}{|\mathcal{X}^{(1)}||\mathcal{X}^{(2)}||\mathcal{\widehat{Y}}^{(1)}||\mathcal{\widehat{Y}}^{(2)}|})$.
Here we used the fact (discussed above) that the
total variation distance between the joint distribution of
$(X^{(1)}, X^{(2)}, \widehat{Y}^{(1)}, \widehat{Y}^{(2)})$ and that
of $(X^{(1)}, X^{(2)}, Y^{(1)}, Y^{(2)})$ is less than or equal to
$\epsilon$.
We bound the second term, i.e.
$-\frac{1}{n}H(X^{(1)}_{1:n}\widehat{Y}^{(1)}_{1:n}\widehat{Y}^{(2)}_{1:n}|X^{(2)}_{1:n}C)$,
from below as follows:
\begin{eqnarray*}&-\frac{1}{n}H(X^{(1)}_{1:n}\widehat{Y}^{(1)}_{1:n}\widehat{Y}^{(2)}_{1:n}|X^{(2)}_{1:n}C)=
-\frac{1}{n}\sum_{j=1}^nH(X^{(1)}_{j}\widehat{Y}^{(1)}_{j}\widehat{Y}^{(2)}_{j}|X^{(2)}_{1:n}CX^{(1)}_{1:j-1}\widehat{Y}^{(1)}_{1:j-1}\widehat{Y}^{(2)}_{1:j-1})\geq\\&
-\frac{1}{n}\sum_{j=1}^nH(X^{(1)}_{j}\widehat{Y}^{(1)}_{j}\widehat{Y}^{(2)}_{j}|X^{(2)}_{j}X^{(2)}_{1:j-1}X^{(2)}_{j+1:n}C)=H(X^{(1)}\widehat{Y}^{(1)}\widehat{Y}^{(2)}|X^{(2)}U)\end{eqnarray*}
Therefore
\begin{eqnarray*}&R_{12}\geq
I(X^{(1)}\widehat{Y}^{(1)}\widehat{Y}^{(2)};U|X^{(2)})-2\kappa(\epsilon)\end{eqnarray*}
Hence, $(R_{12},0)$ belongs to $\mathcal{R}^{\epsilon}_{1}$
defined as follows:
\begin{eqnarray*}\big\{(R_{12},0)\in \mathcal{R}^2&:& \exists
p(u, x^{(1)}, x^{(2)}, \widehat{y}^{(1)}, \widehat{y}^{(2)})\in
T^{\epsilon}_1 \mbox{ s.t. }\\R_{12}&\geq&
I(U;X^{(1)}\widehat{Y}^{(1)}\widehat{Y}^{(2)}|X^{(2)})-2\kappa(\epsilon)\big\},
\end{eqnarray*}
where
\begin{eqnarray*}
T^{\epsilon}_1\triangleq\big\{p(u, x^{(1)}, x^{(2)},
\widehat{y}^{(1)}, \widehat{y}^{(2)})&:& \|p(x^{(1)}, x^{(2)},
\widehat{y}^{(1)}, \widehat{y}^{(2)})-q(x^{(1)}, x^{(2)},
\widehat{y}^{(1)}, \widehat{y}^{(2)})\|_1\leq \epsilon\\&&
\widehat{Y}^{(2)}-UX^{(2)}-X^{(1)}\widehat{Y}^{(1)}\\&&|\mathcal{U}|\leq
|\mathcal{X}^{(1)}||\mathcal{\widehat{Y}}^{(1)}||\mathcal{X}^{(2)}||\mathcal{\widehat{Y}}^{(2)}|+1\big\},
\end{eqnarray*}
since a cardinality bound on $|\mathcal{U}|$ can be imposed using
the generalized Carath\'{e}odory theorem of Fenchel. 

The proof will be done by noting that
$\bigcap_{\epsilon>0}\mathcal{R}^{\epsilon}_{1}=\mathcal{R}_{1}$.
In order to show the latter, note that $\mathcal{R}_{1}\subset
\bigcap_{\epsilon>0}\mathcal{R}^{\epsilon}_{1}$ by definition. In
order to show
$\bigcap_{\epsilon>0}\mathcal{R}^{\epsilon}_{1}\subset
\mathcal{R}_{1}$, we take a sequence of $\epsilon_1$,
$\epsilon_2$,... where $\lim_{i\rightarrow \infty}\epsilon_i=0$.
Take a point in $\bigcap_{i\geq 1}\mathcal{R}^{\epsilon_i}_{1}$.
Corresponding to this point are $p_i(u, x^{(1)}, x^{(2)},
\widehat{y}^{(1)}, \widehat{y}^{(2)})$ in $T^{\epsilon_i}_1$. Since
we have cardinality bounds on the alphabet of the random variables
involved, one can think of these probability distributions as
vectors in the probability simplex of
$\mathcal{R}^{|\mathcal{U}||\mathcal{X}^{(1)}||\mathcal{\widehat{Y}}^{(1)}||\mathcal{X}^{(2)}||\mathcal{\widehat{Y}}^{(2)}|}$.
Since the probability simplex is a compact set (when viewed as a
subset of the Euclidean space), there must exist a subsequence
$i_1, i_2, ...$ where the sequence $p_{i_k}(u, x^{(1)}, x^{(2)},
\widehat{y}^{(1)}, \widehat{y}^{(2)})$ for $k=1,2,3,...$ converges
to some $p^*(u, x^{(1)}, x^{(2)}, \widehat{y}^{(1)},
\widehat{y}^{(2)})$. $p^*(u, x^{(1)}, x^{(2)}, \widehat{y}^{(1)},
\widehat{y}^{(2)})$ must belong to $T_1$ using the fact that mutual
information function and the total variational distance are
continuous in the underlying joint distribution. One can further
observe that the point in $\bigcap_{i\geq
1}\mathcal{R}^{\epsilon_i}_{1}$ that we started with, also belongs
to $\mathcal{R}_{1}$ for the choice of $p^*(u, x^{(1)}, x^{(2)},
\widehat{y}^{(1)}, \widehat{y}^{(2)})$ since $\lim_{i\rightarrow
\infty}2\kappa(\epsilon_i)$ is zero.
\end{proof}

\begin{proof}[Proof of Theorem \ref{Thm:2}] Take an arbitrary $q(x^{(1)}, x^{(2)}, y^{(1)}, y^{(2)})$. Since
$(R_{12}, R_{21})$ is admissible with $r$ rounds of communication,
for any positive $\epsilon$, a code of length $n$ exists such that
the communication rates are less than or equal to $(R_{12}, R_{21})$
and the total variation distance between the overall reconstruction and the
distribution of $n$ i.i.d copies of $(X^{(1)}, X^{(2)}, Y^{(1)},
Y^{(2)})$ is less than or equal to $\epsilon$. In other words for
some natural number $n$, the two nodes observe $n$ i.i.d. copies
of $X^{(1)}, X^{(2)}$, i.e. $X^{(1)}_{1:n}, X^{(2)}_{1:n}$ respectively. Assume
that the first node is using the external randomness $M_1$, and the
second node is using the external randomness $M_2$. $M_1$, $M_2$
and $X^{(1)}_{1:n}X^{(2)}_{1:n}$ must be mutually independent. Let
$C_1, C_2, ..., C_r$ denote the interactive communication used in
the code. We have $H(C_i|C_{1:i-1}X^{(1)}_{1:n}M_1)=0$ if $i$ is
odd, $H(C_i|C_{1:i-1}X^{(2)}_{1:n}M_2)=0$ if $i$ is even. We further
have
$$\frac{1}{n}\sum_{i: odd}H(C_i)\leq R_{12}$$
and
$$\frac{1}{n}\sum_{i: even}H(C_i)\leq R_{21}.$$
Following the communication, the first node creates
$\widehat{Y}^{(1)}_{1:n}$, and the second node creates
$\widehat{Y}^{(2)}_{1:n}$ such that the total variation distance
between $p(X^{(1)}_{1:n}, X^{(2)}_{1:n}, \widehat{Y}^{(1)}_{1:n},
\widehat{Y}^{(2)}_{1:n})$ and the distribution of $(X_{1:n}^{(1)},
X_{1:n}^{(2)}, Y_{1:n}^{(1)}, Y_{1:n}^{(2)})$, constructed by taking
$n$ i.i.d copies of $q(x^{(1)}, x^{(2)}, y^{(1)}, y^{(2)})$, is less
than or equal to $\epsilon$.

Take a random variable $J$ uniform on $\{1,2,3,...,n\}$ and
independent of all the above mentioned random variables. Let
\begin{eqnarray*}&Z=JX^{(2)}_{1:J-1}X^{(1)}_{J+1:n},\\&F_i=C_i\ (\mbox{for } i=1,2,3,...,r) and \overrightarrow{F}=(F_1, F_2, ..., F_r),\\&
X^{(1)}=X^{(1)}_J, X^{(2)}=X^{(2)}_J,
\widehat{Y}^{(1)}=\widehat{Y}^{(1)}_J,
\widehat{Y}^{(2)}=\widehat{Y}^{(2)}_J\end{eqnarray*}

The following statements hold:
\begin{enumerate}
  \item $I(\widehat{Y}^{(1)};X^{(2)}\widehat{Y}^{(2)}|X^{(1)}\overrightarrow{F}Z)=0$
  since, for any $1\leq j\leq n$, we have:
\begin{eqnarray*}&0=I(X^{(1)}_{1:n}M_1;X^{(2)}_{1:n}M_2|X^{(1)}_{j:n}X^{(2)}_{1:j-1})=I(X^{(1)}_{1:n}M_1F_1;X^{(2)}_{1:n}M_2|X^{(1)}_{j:n}X^{(2)}_{1:j-1})\geq\\&
I(X^{(1)}_{1:n}M_1;X^{(2)}_{1:n}M_2|X^{(1)}_{j:n}X^{(2)}_{1:j-1}F_1)=I(X^{(1)}_{1:n}M_1;X^{(2)}_{1:n}M_2F_2|X^{(1)}_{j:n}X^{(2)}_{1:j-1}F_1)\geq\\&
I(X^{(1)}_{1:n}M_1;X^{(2)}_{1:n}M_2|X^{(1)}_{j:n}X^{(2)}_{1:j-1}F_{1:2})\geq
\cdot\cdot\cdot \geq
I(X^{(1)}_{1:n}M_1;X^{(2)}_{1:n}M_2|X^{(1)}_{j:n}X^{(2)}_{1:j-1}\overrightarrow{F})\geq\\&
I(\widehat{Y}^{(1)}_{1:n};X^{(2)}_{1:n}\widehat{Y}^{(2)}_{1:n}|X^{(1)}_{j:n}X^{(2)}_{1:j-1}\overrightarrow{F})\geq
I(\widehat{Y}^{(1)}_{j};X^{(2)}_{j}\widehat{Y}^{(2)}_{j}|X^{(1)}_{j}X^{(1)}_{j+1:n}X^{(2)}_{1:j-1}\overrightarrow{F}).
\end{eqnarray*}
  Therefore
  $I(\widehat{Y}^{(1)};X^{(2)}\widehat{Y}^{(2)}|X^{(1)}\overrightarrow{F}Z)=0$.
  Note that this equation implies that for any value of $Z=z$,
  $I(\widehat{Y}^{(1)};X^{(2)}\widehat{Y}^{(2)}|X^{(1)}\overrightarrow{F},Z=z)=0$.
  \item Similarly one can show that
  $I(\widehat{Y}^{(2)};X^{(1)}\widehat{Y}^{(1)}|X^{(2)}\overrightarrow{F}Z)=0$.
  Note that this equation implies that for any value of $Z=z$,
  $I(\widehat{Y}^{(2)};X^{(1)}\widehat{Y}^{(1)}|X^{(2)}\overrightarrow{F},Z=z)=0$.
  \item $I(F_i;X^{(2)}|F_{1:i-1}X^{(1)}Z)=0$ if $i$ is odd, since following
  the beginning steps of the proof for the first statement, one can
  show that for any $1\leq j\leq n$ the following inequality holds:
\begin{eqnarray*}&0\geq I(X^{(1)}_{1:n}M_1;X^{(2)}_{1:n}M_2|X^{(1)}_{j:n}X^{(2)}_{1:j-1}F_{1:i-1})\end{eqnarray*}
  Next, note that
\begin{eqnarray*}&I(X^{(1)}_{1:n}M_1;X^{(2)}_{1:n}M_2|X^{(1)}_{j:n}X^{(2)}_{1:j-1}F_{1:i-1})\geq
I(F_i;X^{(2)}_{1:n}|X^{(1)}_{j:n}X^{(2)}_{1:j-1}F_{1:i-1})\geq\\&
I(F_i;X^{(2)}_{j}|X^{(1)}_{j:n}X^{(2)}_{1:j-1}F_{1:i-1})=I(F_i;X^{(2)}_{j}|X^{(1)}_jX^{(1)}_{j+1:n}X^{(2)}_{1:j-1}F_{1:i-1}).
\end{eqnarray*}
Therefore $I(F_i;X^{(2)}|F_{1:i-1}X^{(1)}Z)=0$ if $i$ is odd. Note
that this equation implies that for any value of $Z=z$,
$I(F_i;X^{(2)}|F_{1:i-1}X^{(1)}, Z=z)=0$ if $i$ is odd.
  \item Similarly one can show that $I(F_i;X^{(1)}|F_{1:i-1}X^{(2)}Z)=0$ if $i$ is
  even. Note that this equation implies that for any value of $Z=z$,
  $I(F_i;X^{(1)}|F_{1:i-1}X^{(2)},Z=z)=0$ if $i$ is
  even.
\end{enumerate}
These statements imply that conditioned on any $Z=z$, the
conditional distribution \begin{eqnarray*}&p(f_1, f_2, ..., f_r,
x^{(1)}, x^{(2)}, \widehat{y}^{(1)},
\widehat{y}^{(2)}|z)\end{eqnarray*} satisfies following Markov
chains equations:
\begin{eqnarray*}&
F_i-F_{1:i-1}X^{(1)}-X^{(2)} \mbox{ if } i \mbox{ is odd},\\&
F_i-F_{1:i-1}X^{(2)}-X^{(1)} \mbox{ if } i \mbox{
is even}\\& \widehat{Y}^{(1)}-\overrightarrow{F}X^{(1)}-X^{(2)}\widehat{Y}^{(2)}\\
&\widehat{Y}^{(2)}-\overrightarrow{F}X^{(2)}-X^{(1)}\widehat{Y}^{(1)}
\end{eqnarray*}

Let $d_z$ denote the total variation distance between
$p(X^{(1)}=x^{(1)}, X^{(2)}=x^{(2)}, \widehat{Y}^{(1)}=y^{(1)},
\widehat{Y}^{(2)}=y^{(2)}|Z=z)$ and $p(X^{(1)}=x^{(1)},
X^{(2)}=x^{(2)}, Y^{(1)}=y^{(1)},
  y^{(2)}=y^{(2)})$. In appendix \ref{AppendixA}, we show that $\sum_z p(z)d_z\leq \epsilon$.

In appendix \ref{AppendixB}, we show that the triple $\big(R_{12},
R_{21},
R_{12}+R_{21}-I(Y^{(1)};Y^{(2)}|X^{(1)}X^{(2)})+3\kappa(\epsilon)\big)$
is coordinate by coordinate greater than or equal to
\begin{eqnarray*}&\sum_z p(z)\bigg(I(\overrightarrow{F};X^{(1)}|X^{(2)}, Z=z),
I(\overrightarrow{F};X^{(2)}|X^{(1)}, Z=z),\\&
I(\overrightarrow{F};X^{(1)}|X^{(2)}, Z=z)+
I(\overrightarrow{F};X^{(2)}|X^{(1)}, Z=z)\bigg)\end{eqnarray*}
where
$\kappa(\epsilon)=-\epsilon\log(\frac{\epsilon}{|\mathcal{X}^{(1)}||\mathcal{X}^{(2)}||\mathcal{\widehat{Y}}^{(1)}||\mathcal{\widehat{Y}}^{(2)}|})$
and $I(Y^{(1)};Y^{(2)}|X^{(1)}X^{(2)})$ is evaluated assuming that
$\\X^{(1)}, X^{(2)}, Y^{(1)}, Y^{(2)}$ have the joint distribution
of $q(x^{(1)}, x^{(2)}, y^{(1)}, y^{(2)})$.

Therefore the four-tuple $\big(R_{12}, R_{21},
R_{12}+R_{21}-I(Y^{(1)};Y^{(2)}|X^{(1)}X^{(2)})+3\kappa(\epsilon),
\epsilon\big)$ is coordinate by coordinate greater than or equal to
\begin{eqnarray*}&\sum_z p(z)\bigg(I(\overrightarrow{F};X^{(1)}|X^{(2)}, Z=z),
I(\overrightarrow{F};X^{(2)}|X^{(1)}, Z=z),\\&
I(\overrightarrow{F};X^{(1)}|X^{(2)}, Z=z)+
I(\overrightarrow{F};X^{(2)}|X^{(1)}, Z=z),
d_z\bigg).\end{eqnarray*} The Carath\'{e}odory theorem implies that the
four tuple $\big(R_{12}, R_{21},
R_{12}+R_{21}-I(Y^{(1)};Y^{(2)}|X^{(1)}X^{(2)})+3\kappa(\epsilon),
\epsilon\big)$ is coordinate by coordinate greater than or equal to
\begin{eqnarray*}&\sum_{i=1:5} r_i\bigg(I(\overrightarrow{F};X^{(1)}|X^{(2)}, Z=z_i),
I(\overrightarrow{F};X^{(2)}|X^{(1)}, Z=z_i),\\&
I(\overrightarrow{F};X^{(1)}|X^{(2)}, Z=z_i)+
I(\overrightarrow{F};X^{(2)}|X^{(1)}, Z=z_i),
d_{z_i}\bigg).\end{eqnarray*} for some values of $z_1$, $z_2$,...,
$z_5$ and non-negative $r_1$, $r_2$,...,$r_5$ where $\sum_{i=1:5}
r_i=1$. Since for $i=1,...,5$, we have $r_id_{z_i}\leq \epsilon$, if
$d_{z_i}\geq \sqrt\epsilon$, $r_i$ will be less than or equal to
$\sqrt\epsilon$. Let $t=\sum_{i:d_{z_i}\leq \sqrt\epsilon}r_i$. Note
that
$1-5\sqrt\epsilon \leq t\leq 1$. 
We then have:
\begin{eqnarray*}&\sum_{i=1:5} r_i\bigg(I(\overrightarrow{F};X^{(1)}|X^{(2)}, Z=z_i),
I(\overrightarrow{F};X^{(2)}|X^{(1)}, Z=z_i),\\&
I(\overrightarrow{F};X^{(1)}|X^{(2)}, Z=z_i)+
I(\overrightarrow{F};X^{(2)}|X^{(1)}, Z=z_i),
d_{z_i}\bigg).\end{eqnarray*} is coordinate by coordinate greater
than or equal to
\begin{eqnarray*}&(1-5\sqrt\epsilon)\cdot\sum_{i:d_{z_i}\leq \sqrt\epsilon}\ \frac{r_i}{t}\bigg(I(\overrightarrow{F};X^{(1)}|X^{(2)}, Z=z_i),
I(\overrightarrow{F};X^{(2)}|X^{(1)}, Z=z_i),\\&
I(\overrightarrow{F};X^{(1)}|X^{(2)}, Z=z_i)+
I(\overrightarrow{F};X^{(2)}|X^{(1)}, Z=z_i),
d_{z_i}\bigg)\end{eqnarray*} Therefore $\big(R_{12}, R_{21},
R_{12}+R_{21}-I(Y^{(1)};Y^{(2)}|X^{(1)}X^{(2)})+3\kappa(\epsilon)\big)$
is coordinate by coordinate greater than or equal to
\begin{eqnarray*}&(1-5\sqrt\epsilon)\sum_{i:d_{z_i}\leq \sqrt\epsilon}\ \frac{r_i}{t}\bigg(I(\overrightarrow{F};X^{(1)}|X^{(2)}, Z=z_i),
I(\overrightarrow{F};X^{(2)}|X^{(1)}, Z=z_i),\\&
I(\overrightarrow{F};X^{(1)}|X^{(2)}, Z=z_i)+
I(\overrightarrow{F};X^{(2)}|X^{(1)}, Z=z_i)\bigg)\end{eqnarray*}

Therefore any admissible pair $(R_{12}, R_{21})$ must belong to
$\mathcal{R}^{\epsilon}_{2}(r)$ defined as the convex closure of
\begin{eqnarray*}\big\{(R_{12}, R_{21})\in \mathcal{R}^2&:& \exists
p(f_1, f_2, ..., f_r, x^{(1)}, x^{(2)}, \widehat{y}^{(1)},
\widehat{y}^{(2)})\in T^{\epsilon}_2(r),\\&& (X^{(1)}, X^{(2)},
Y^{(1)}, Y^{(2)}) \sim q(x^{(1)}, x^{(2)}, y^{(1)}, y^{(2)}) \mbox{
s.t. }\\R_{12}&\geq& (1-5\sqrt\epsilon)I(X^{(1)};
\overrightarrow{F}|X^{(2)}),\\R_{21}&\geq&(1-5\sqrt\epsilon)I(X^{(2)};
\overrightarrow{F}|X^{(1)})\\
R_{12}+R_{21}&\geq& (1-5\sqrt\epsilon)I(X^{(1)};
\overrightarrow{F}|X^{(2)})+(1-5\sqrt\epsilon)I(X^{(2)};
\overrightarrow{F}|
X^{(1)})+\\&&+I(Y^{(1)};Y^{(2)}|X^{(1)}X^{(2)})-3\kappa(\epsilon)\big\},
\end{eqnarray*}
where
\begin{eqnarray*}
T^{\epsilon}_2(r)\triangleq\big\{p(f_1, f_2, ..., f_r, x^{(1)},
x^{(2)}, \widehat{y}^{(1)}, \widehat{y}^{(2)})&:& \|p(x^{(1)},
x^{(2)}, \widehat{y}^{(1)}, \widehat{y}^{(2)})-q(x^{(1)}, x^{(2)},
\widehat{y}^{(1)}, \widehat{y}^{(2)})\|_1\leq \sqrt\epsilon\\&&
F_i-F_{1:i-1}X^{(1)}-X^{(2)} \mbox{ if } i \mbox{ is odd},\\&&
F_i-F_{1:i-1}X^{(2)}-X^{(1)} \mbox{ if } i \mbox{
is even}\\&& \widehat{Y}^{(1)}-\overrightarrow{F}X^{(1)}-X^{(2)}\widehat{Y}^{(2)}\\
&&\forall\ 1\leq i\leq r:\\&&\ \ \ |\mathcal{F}_i|\leq
|\mathcal{F}_{i-1}||\mathcal{F}_{i-2}|\cdot\cdot\cdot
|\mathcal{F}_{1}||\mathcal{X}^{(1)}||\mathcal{Y}^{(1)}||\mathcal{X}^{(2)}||\mathcal{Y}^{(2)}|+1\big\}.
\end{eqnarray*}
since cardinality bounds on $|\mathcal{F}_i|$ can be imposed using
the generalized Carath\'{e}odory theorem of Fenchel.
This is argued in appendix \ref{AppendixC}. We can then continue as
in the proof of theorem \ref{Thm:1} to show that
$\bigcap_{\epsilon>0}\mathcal{R}^{\epsilon}_{2}(r)=\mathcal{R}_{2}(r)$,
and hence $(R_{12}, R_{21})$ must belong to
$\mathcal{R}^{\epsilon}_{2}(r)$.
\end{proof}

\begin{proof}[Proof of Theorem \ref{Thm:3}]
It suffices to show that $\bigcup_{r\geq
1}\mathcal{R}_{2}(r)\subset \mathcal{R}_{3}$. Take some $r$ and
some point $(R_{12}, R_{21})\in \mathcal{R}_{2}(r)$. Corresponding
to this point is some $p(f_1, f_2, ..., f_r, x^{(1)}, x^{(2)},
y^{(1)}, y^{(2)})\in T_2(r)$. Let $U=F_{1:r}$ be jointly distributed
with $X^{(1)}, X^{(2)}, Y^{(1)}, Y^{(2)}$ according to $p((f_1, f_2,
..., f_r), x^{(1)}, x^{(2)}, y^{(1)}, y^{(2)})$. Note that
$I(X^{(1)}; X^{(2)}|U)\leq I(X^{(1)}; X^{(2)})$ since
\begin{eqnarray*}&I(X^{(1)}; X^{(2)})=I(X^{(1)}F_1; X^{(2)})\geq
I(X^{(1)}; X^{(2)}|F_1)=I(X^{(1)}; X^{(2)}F_2|F_1)\geq\\& I(X^{(1)};
X^{(2)}|F_{1:2})\geq \cdot\cdot\cdot \geq I(X^{(1)};
X^{(2)}|F_{1:r})=I(X^{(1)}; X^{(2)}|U),\end{eqnarray*} where we have
used the fact that for any $p(f_1, f_2, ..., f_r, x^{(1)}, x^{(2)},
y^{(1)}, y^{(2)})\in T_2(r)$, the following Markov chain equations
hold:\begin{eqnarray*}&F_i-F_{1:i-1}X^{(1)}-X^{(2)} \mbox{ if } i
\mbox{ is odd},\\& F_i-F_{1:i-1}X^{(2)}-X^{(1)} \mbox{ if } i \mbox{
is even}\end{eqnarray*}

One can verify that $p(u, x^{(1)}, x^{(2)}, y^{(1)}, y^{(2)})$
satisfies the required properties for being in $T_3$, except the
cardinality bound. Similar to what was done at the end of the proof
of theorem \ref{Thm:2}, one can use the generalized Carath\'{e}odory
theorem of Fenchel to impose the desired cardinality
bound $|\mathcal{U}|$.
\end{proof}

\begin{proof}[Proof of Theorem \ref{Thm:Achievability1}] It suffices to prove
the generalized lemma stated in the discussion following the statement of the theorem.
Here we build upon the ideas Paul Cuff uses to prove the direct part for the problem he considers.

We denote the induced joint distribution on $(x_{1:n}, y^{(1)}_{1:n}, y^{(2)}_{1:n})$ by $\hat{P}(x_{1:n}, y^{(1)}_{1:n}, y^{(2)}_{1:n})$. We use the capital letter $\hat{P}$ to indicate that this probability is a random variable due to the random generation of the codebook. The induced joint distribution on $(x_{1:n}, y^{(1)}_{1:n}, y^{(2)}_{1:n})$ is equal to:
$$\hat{P}(x_{1:n}, y^{(1)}_{1:n}, y^{(2)}_{1:n})=$$$$\sum_{1\leq t_i\leq 2^{nR'_i} \mbox{\small~~for~~}i=1,2,..,r}\bigg(
\prod_{i=1}^nq(x_i)2^{-n\sum_{i=1}^r R'_i}\prod_{i=1}^nq(y^{(1)}_{i}|F_{i}^{(1)}(t_1), F_{i}^{(2)}(t_1,t_2), ..., F_{i}^{(r)}(t_1,t_2,t_3,...,t_r),x_{i})$$$$
\prod_{i=1}^nq(y^{(2)}_{i}|F_{i}^{(1)}(t_1), F_{i}^{(2)}(t_1,t_2), ..., F_{i}^{(r)}(t_1,t_2,t_3,...,t_r),x_{i})\bigg).$$
To prove the theorem, it suffices to show that the expected value of the total variance between the random variable $\hat{P}(x_{1:n}, y^{(1)}_{1:n}, y^{(2)}_{1:n})$ and $\prod_{i=1}^nq(x_i, y^{(1)}_{i}, y^{(2)}_{i})$ converges to zero as $n$ converges infinity, i.e.
$$\lim_{n\rightarrow \infty}\sum_{x_{1:n}, y^{(1)}_{1:n}, y^{(2)}_{1:n}}\textbf{E}\big|\hat{P}(x_{1:n}, y^{(1)}_{1:n}, y^{(2)}_{1:n})-\prod_{i=1}^nq(x_i, y^{(1)}_{i}, y^{(2)}_{i})\big| = 0.$$
Using the linearity of expectation, we show that $\textbf{E}[\hat{P}(x_{1:n}, y^{(1)}_{1:n}, y^{(2)}_{1:n})]=\prod_{i=1}^nq(x_i, y^{(1)}_{i}, y^{(2)}_{i})$ as follows:
$$\textbf{E}[\hat{P}(x_{1:n}, y^{(1)}_{1:n}, y^{(2)}_{1:n})]=$$$$2^{n\sum_{i=1}^rR'_i}\textbf{E}\bigg[\prod_{i=1}^nq(x_i)2^{-n\sum_{i=1}^r R'_i}\prod_{i=1}^nq(y^{(1)}_{i}|F_{i}^{(1)}(1), F_{i}^{(2)}(1,1), ..., F_{i}^{(r)}(1,1,1,...,1),x_{i})=$$$$
\prod_{i=1}^nq(y^{(2)}_{i}|F_{i}^{(1)}(1), F_{i}^{(2)}(1,1), ..., F_{i}^{(r)}(1,1,1,...,1),x_{i})\bigg]$$
$$\prod_{i=1}^nq(x_i)\textbf{E}\bigg[\prod_{i=1}^nq(y^{(1)}_{i}|F_{i}^{(1)}(1), F_{i}^{(2)}(1,1), ..., F_{i}^{(r)}(1,1,1,...,1),x_{i})$$$$
\prod_{i=1}^nq(y^{(2)}_{i}|F_{i}^{(1)}(1), F_{i}^{(2)}(1,1), ..., F_{i}^{(r)}(1,1,1,...,1),x_{i})\bigg]$$
$$=\prod_{i=1}^nq(x_i)\sum_{f_{1:n}^{(1)}(1),f_{1:n}^{(2)}(1,1),..., f_{1:n}^{(r)}(1,1,1,...,1)}\bigg[\prod_{i=1}^nq\big(f_{i}^{(1)}(1),f_{i}^{(2)}(1,1),..., f_{i}^{(r)}(1,1,1,...,1)|x_i\big)\cdot$$$$\prod_{i=1}^nq(y^{(1)}_{i},y^{(2)}_{i}|f_{i}^{(1)}(1), f_{i}^{(2)}(1,1), ..., f_{i}^{(r)}(1,1,1,...,1),x_{i})\bigg]$$
$$=\prod_{i=1}^nq(x_i)\prod_{i=1}^nq(y^{(1)}_{i}, y^{(2)}_{i}|x_i)=\prod_{i=1}^nq(x_i, y^{(1)}_{i}, y^{(2)}_{i}).$$
Take some arbitrary $\epsilon>0$. We first take care of the sum of the total variance over those of non-typical sequences $\mathcal{T}_{\epsilon}^{(n)}$. Using the inequality $\textbf{E}\big[|X-\textbf{E}[X]|\big]\leq 2\textbf{E}[X]$ for non-negative random variables, we can write
$$\sum_{(x_{1:n}, y^{(1)}_{1:n}, y^{(2)}_{1:n})\notin \mathcal{T}_{\epsilon}^{(n)}}\textbf{E}\big|\hat{P}(x_{1:n}, y^{(1)}_{1:n}, y^{(2)}_{1:n})-\prod_{i=1}^nq(x_i, y^{(1)}_{i}, y^{(2)}_{i})\big|\leq$$$$\sum_{(x_{1:n}, y^{(1)}_{1:n}, y^{(2)}_{1:n})\notin \mathcal{T}_{\epsilon}^{(n)}}2\textbf{E}\hat{P}(x_{1:n}, y^{(1)}_{1:n}, y^{(2)}_{1:n})
=$$$$\sum_{(x_{1:n}, y^{(1)}_{1:n}, y^{(2)}_{1:n})\notin \mathcal{T}_{\epsilon}^{(n)}}2\prod_{i=1}^nq(x_i, y^{(1)}_{i}, y^{(2)}_{i})\leq 2\epsilon.$$
It remains to bound $$\sum_{(x_{1:n}, y^{(1)}_{1:n}, y^{(2)}_{1:n})\in \mathcal{T}_{\epsilon}^{(n)}}\textbf{E}\big|\hat{P}(x_{1:n}, y^{(1)}_{1:n}, y^{(2)}_{1:n})-\prod_{i=1}^nq(x_i, y^{(1)}_{i}, y^{(2)}_{i})\big|.$$ Let us define $\hat{P}_1(x_{1:n}, y^{(1)}_{1:n}, y^{(2)}_{1:n})$ and $\hat{P}_2(x_{1:n}, y^{(1)}_{1:n}, y^{(2)}_{1:n})$ as follows:
$$\hat{P}_1(x_{1:n}, y^{(1)}_{1:n}, y^{(2)}_{1:n})=$$$$\sum_{1\leq t_i\leq 2^{nR'_i} \mbox{\small~~for~~}i=1,2,..,r}\bigg(
\prod_{i=1}^nq(x_i)2^{-n\sum_{i=1}^r R'_i}\prod_{i=1}^nq(y^{(1)}_{i},y^{(2)}_{i}|F_{i}^{(1)}(t_1), F_{i}^{(2)}(t_1,t_2), ..., F_{i}^{(r)}(t_1,t_2,t_3,...,t_r),x_{i})$$$$\mathbf{1}\big[(y^{(1)}_{1:n},y^{(2)}_{1:n},F_{1:n}^{(1)}(1), ...,F_{1:n}^{(r)}(1,1,1,...,1),x_{1:n})\in \mathcal{T}_{\epsilon}^{(n)}(q(
y^{(1)}, y^{(2)}, f^{(1)}, f^{(2)}, ..., f^{(r)},x))~\big]\bigg),$$
$$\hat{P}_2(x_{1:n}, y^{(1)}_{1:n}, y^{(2)}_{1:n})=$$$$\sum_{1\leq t_i\leq 2^{nR'_i} \mbox{\small~~for~~}i=1,2,..,r}\bigg(
\prod_{i=1}^nq(x_i)2^{-n\sum_{i=1}^r R'_i}\prod_{i=1}^nq(y^{(1)}_{i},y^{(2)}_{i}|F_{i}^{(1)}(t_1), F_{i}^{(2)}(t_1,t_2), ..., F_{i}^{(r)}(t_1,t_2,t_3,...,t_r),x_{i})$$$$\mathbf{1}\big[(y^{(1)}_{1:n},y^{(2)}_{1:n},F_{1:n}^{(1)}(1), ...,F_{1:n}^{(r)}(1,1,1,...,1),x_{1:n})\notin \mathcal{T}_{\epsilon}^{(n)}(q(
y^{(1)}, y^{(2)}, f^{(1)}, f^{(2)}, ..., f^{(r)},x))~\big]\bigg).$$
Using the triangle inequality and the fact that $\hat{P}(x_{1:n}, y^{(1)}_{1:n}, y^{(2)}_{1:n})=\hat{P}_1(x_{1:n}, y^{(1)}_{1:n}, y^{(2)}_{1:n})+\hat{P}_2(x_{1:n}, y^{(1)}_{1:n}, y^{(2)}_{1:n})$ we can write
$$\textbf{E}\big|\hat{P}(x_{1:n}, y^{(1)}_{1:n}, y^{(2)}_{1:n})-\textbf{E}\hat{P}(x_{1:n}, y^{(1)}_{1:n}, y^{(2)}_{1:n})\big|\leq$$
$$\textbf{E}\big|\hat{P}_1(x_{1:n}, y^{(1)}_{1:n}, y^{(2)}_{1:n})-\textbf{E}\hat{P}_1(x_{1:n}, y^{(1)}_{1:n}, y^{(2)}_{1:n})\big|+$$
$$\textbf{E}\big|\hat{P}_2(x_{1:n}, y^{(1)}_{1:n}, y^{(2)}_{1:n})-\textbf{E}\hat{P}_2(x_{1:n}, y^{(1)}_{1:n}, y^{(2)}_{1:n})\big|.$$
Therefore
$$\sum_{(x_{1:n}, y^{(1)}_{1:n}, y^{(2)}_{1:n})\in \mathcal{T}_{\epsilon}^{(n)}}\textbf{E}\big|\hat{P}(x_{1:n}, y^{(1)}_{1:n}, y^{(2)}_{1:n})-\prod_{i=1}^nq(x_i, y^{(1)}_{i}, y^{(2)}_{i})\big|\leq$$
$$\sum_{(x_{1:n}, y^{(1)}_{1:n}, y^{(2)}_{1:n})\in \mathcal{T}_{\epsilon}^{(n)}}\textbf{E}\big|\hat{P}_1(x_{1:n}, y^{(1)}_{1:n}, y^{(2)}_{1:n})-\textbf{E}\hat{P}_1(x_{1:n}, y^{(1)}_{1:n}, y^{(2)}_{1:n})\big|+$$
$$\sum_{(x_{1:n}, y^{(1)}_{1:n}, y^{(2)}_{1:n})\in \mathcal{T}_{\epsilon}^{(n)}}\textbf{E}\big|\hat{P}_2(x_{1:n}, y^{(1)}_{1:n}, y^{(2)}_{1:n})-\textbf{E}\hat{P}_2(x_{1:n}, y^{(1)}_{1:n}, y^{(2)}_{1:n})\big|.$$
Let us begin with the second term
$$\sum_{(x_{1:n}, y^{(1)}_{1:n}, y^{(2)}_{1:n})\in \mathcal{T}_{\epsilon}^{(n)}}\textbf{E}\big|\hat{P}_2(x_{1:n}, y^{(1)}_{1:n}, y^{(2)}_{1:n})-\textbf{E}\hat{P}_2(x_{1:n}, y^{(1)}_{1:n}, y^{(2)}_{1:n})\big|\leq$$
$$\sum_{(x_{1:n}, y^{(1)}_{1:n}, y^{(2)}_{1:n})\in \mathcal{T}_{\epsilon}^{(n)}}2\textbf{E}\hat{P}_2(x_{1:n}, y^{(1)}_{1:n}, y^{(2)}_{1:n})\leq$$
$$2\sum_{(x_{1:n}, y^{(1)}_{1:n}, y^{(2)}_{1:n})\in \mathcal{T}_{\epsilon}^{(n)}}\prod_{i=1}^nq(x_i)\textbf{E}\bigg(\prod_{i=1}^nq(y^{(1)}_{i},y^{(2)}_{i}|F_{i}^{(1)}(t_1), F_{i}^{(2)}(t_1,t_2), ..., F_{i}^{(r)}(t_1,t_2,t_3,...,t_r),x_{i})$$$$\mathbf{1}\big[(y^{(1)}_{1:n},y^{(2)}_{1:n},F_{1:n}^{(1)}(1), ...,F_{1:n}^{(r)}(1,1,1,...,1),x_{1:n})\notin \mathcal{T}_{\epsilon}^{(n)}(q(
y^{(1)}, y^{(2)}, f^{(1)}, f^{(2)}, ..., f^{(r)},x))~\big]\bigg)\leq$$
$$2\sum_{(y^{(1)}_{1:n},y^{(2)}_{1:n},f_{1:n}^{(1)}(1), ...,f_{1:n}^{(r)}(1,1,1,...,1),x_{1:n})\notin \mathcal{T}_{\epsilon}^{(n)}(q(
y^{(1)}, y^{(2)}, f^{(1)}, f^{(2)}, ..., f^{(r)},x))~\big]}\bigg($$$$\prod_{i=1}^nq(x_i)\prod_{i=1}^nq(f_{i}^{(1)}(t_1), f_{i}^{(2)}(t_1,t_2), ..., f_{i}^{(r)}(t_1,t_2,t_3,...,t_r)|x_{i})$$$$\prod_{i=1}^nq(y^{(1)}_{i},y^{(2)}_{i}|f_{i}^{(1)}(t_1), f_{i}^{(2)}(t_1,t_2), ..., f_{i}^{(r)}(t_1,t_2,t_3,...,t_r),x_{i})\bigg)\leq 2\epsilon.$$
We will work out the first term now
$$\sum_{(x_{1:n}, y^{(1)}_{1:n}, y^{(2)}_{1:n})\in \mathcal{T}_{\epsilon}^{(n)}}\textbf{E}\big|\hat{P}_1(x_{1:n}, y^{(1)}_{1:n}, y^{(2)}_{1:n})-\textbf{E}\hat{P}_1(x_{1:n}, y^{(1)}_{1:n}, y^{(2)}_{1:n})\big|\leq$$
$$\sum_{(x_{1:n}, y^{(1)}_{1:n}, y^{(2)}_{1:n})\in \mathcal{T}_{\epsilon}^{(n)}}\sqrt{\textbf{E}\big(\hat{P}_1(x_{1:n}, y^{(1)}_{1:n}, y^{(2)}_{1:n})-\textbf{E}\hat{P}_1(x_{1:n}, y^{(1)}_{1:n}, y^{(2)}_{1:n})\big)^2}=$$
$$\sum_{(x_{1:n}, y^{(1)}_{1:n}, y^{(2)}_{1:n})\in \mathcal{T}_{\epsilon}^{(n)}}\sqrt{\textbf{Var}\big(\hat{P}_1(x_{1:n}, y^{(1)}_{1:n}, y^{(2)}_{1:n})\big)}.$$

We now compute the variance of $\hat{P}_1(x_{1:n}, y^{(1)}_{1:n}, y^{(2)}_{1:n})$ using the formula $Var(K_1+K_2+...+K_m)=\sum_{1\leq i,j\leq m}Cov(K_i, K_j)$:
$$\textbf{Var}\big(\hat{P}_1(x_{1:n}, y^{(1)}_{1:n}, y^{(2)}_{1:n})\big)=$$$$
\sum_{1\leq t_i\leq 2^{nR'_i},1\leq \tilde{t}_i\leq 2^{nR'_i} \mbox{\small~~for~~}i=1,2,..,r}
Cov\bigg($$$$
\prod_{i=1}^nq(x_i)2^{-n\sum_{i=1}^r R'_i}\prod_{i=1}^nq(y^{(1)}_{i},y^{(2)}_{i}|F_{i}^{(1)}(t_1), F_{i}^{(2)}(t_1,t_2), ..., F_{i}^{(r)}(t_1,t_2,t_3,...,t_r),x_{i})$$$$\mathbf{1}\big[(y^{(1)}_{1:n},y^{(2)}_{1:n},F_{1:n}^{(1)}(t_1), ...,F_{1:n}^{(r)}(t_1,t_2,t_3,...,t_r),x_{1:n})\in \mathcal{T}_{\epsilon}^{(n)}(q(
y^{(1)}, y^{(2)}, f^{(1)}, f^{(2)}, ..., f^{(r)},x))~\big]$$$$~~~\textbf{,}\mbox{ and }~~~$$$$
\prod_{i=1}^nq(x_i)2^{-n\sum_{i=1}^r R'_i}\prod_{i=1}^nq(y^{(1)}_{i},y^{(2)}_{i}|F_{i}^{(1)}(\tilde{t}_1), F_{i}^{(2)}(\tilde{t}_1,\tilde{t}_2), ..., F_{i}^{(r)}(\tilde{t}_1,\tilde{t}_2,\tilde{t}_3,...,\tilde{t}_r),x_{i})$$$$\mathbf{1}\big[(y^{(1)}_{1:n},y^{(2)}_{1:n},F_{1:n}^{(1)}(\tilde{t}_1), ...,F_{1:n}^{(r)}(\tilde{t}_1,\tilde{t}_2,\tilde{t}_3,...,\tilde{t}_r),x_{1:n})\in \mathcal{T}_{\epsilon}^{(n)}(q(
y^{(1)}, y^{(2)}, f^{(1)}, f^{(2)}, ..., f^{(r)},x))~\big]\bigg).$$
The number of sequences $(t_1,t_2,...,t_r)$ and $(\tilde{t}_1,\tilde{t}_2,\tilde{t}_3,...,\tilde{t}_r)$ that match in the first $s$ coordinates but differ in the $(s+1)$-th coordinate, i.e. $t_1=\tilde{t}_1, t_2=\tilde{t}_2, ..., t_s=\tilde{t}_s, t_{s+1}\neq \tilde{t}_{s+1}$ is less than or equal to $2^{n\sum_{i=1}^s R'_i}2^{2n\sum_{i=s+1}^r R'_i}$. Given a fixed $s$, the covariance of any two terms corresponding to $(t_1,t_2,...,t_r)$ and $(\tilde{t}_1,\tilde{t}_2,\tilde{t}_3,...,\tilde{t}_r)$ are the same since from the $(s+1)$-th stage onwards the $F$ sequences will fall into different cloud centers and the actual value of the indices $t_{s+2},t_{s+3},...,t_r, \tilde{t}_{s+2},\tilde{t}_{s+3},...,\tilde{t}_r$ have no effect on the joint distribution of the two terms whose covariance is being computed. Therefore the variance of $\hat{P}_1(x_{1:n}, y^{(1)}_{1:n}, y^{(2)}_{1:n})$ can be bounded from above by the summation of $s$ from $0$ and $r$ of
$$2^{n\sum_{i=1}^s R'_i}2^{2n\sum_{i=s+1}^r R'_i}
Cov\bigg($$$$
\prod_{i=1}^nq(x_i)2^{-n\sum_{i=1}^r R'_i}\prod_{i=1}^nq(y^{(1)}_{i},y^{(2)}_{i}|F_{i}^{(1)}(1), F_{i}^{(2)}(1,1), ..., F_{i}^{(r)}(1,1,1,...,1),x_{i})$$$$\mathbf{1}\big[(y^{(1)}_{1:n},y^{(2)}_{1:n},F_{1:n}^{(1)}(1), ...,F_{1:n}^{(r)}(1,1,1,...,1),x_{1:n})\in \mathcal{T}_{\epsilon}^{(n)}(q(
y^{(1)}, y^{(2)}, f^{(1)}, f^{(2)}, ..., f^{(r)},x))~\big]$$$$~~~\textbf{,}\mbox{ and }~~~$$$$
\prod_{i=1}^nq(x_i)2^{-n\sum_{i=1}^r R'_i}\prod_{i=1}^nq(y^{(1)}_{i},y^{(2)}_{i}|F_{i}^{(1)}(1), ...,F_{i}^{(s)}(1,1,1,...,1),F_{i}^{(s+1)}(\underbrace{1,1,1,...1}_{s~\footnotesize\mbox{~terms}},2), F_{i}^{(r)}(\underbrace{1,1,1,...1}_{s~\footnotesize\mbox{~terms}},2,...,2),x_{i})$$
$$\mathbf{1}\big[(y^{(1)}_{1:n},y^{(2)}_{1:n},F_{1:n}^{(1)}(1), ...,F_{i}^{(s)}(1,1,1,...,1),F_{i}^{(s+1)}(\underbrace{1,1,1,...1}_{s~\footnotesize\mbox{~terms}},2), ...,F_{1:n}^{(r)}(\underbrace{1,1,1,...1}_{s~\footnotesize\mbox{~terms}},2,...,2),x_{1:n})\in \mathcal{T}_{\epsilon}^{(n)}~\big]\bigg).$$
The above term is equal to the summation of $s$ from $0$ and $r$ of
$$\prod_{i=1}^nq(x_i)^22^{-n\sum_{i=1}^s R'_i}Cov\bigg($$$$
\prod_{i=1}^nq(y^{(1)}_{i},y^{(2)}_{i}|F_{i}^{(1)}(1), ...,F_{i}^{(s)}(1,1,1,...,1),F_{i}^{(s+1)}(1,1,1,...,1), F_{i}^{(r)}(1,1,1,...,1),x_{i})$$$$\mathbf{1}\big[(y^{(1)}_{1:n},y^{(2)}_{1:n},F_{1:n}^{(1)}(1), ...,F_{i}^{(s)}(1,1,1,...,1),F_{i}^{(s+1)}(1,1,1,...,1), ...,F_{1:n}^{(r)}(1,1,1,...,1),x_{1:n})\in \mathcal{T}_{\epsilon}^{(n)}~\big]$$$$~~~\textbf{,}\mbox{ and }~~~$$$$
\prod_{i=1}^nq(y^{(1)}_{i},y^{(2)}_{i}|F_{i}^{(1)}(1), ...,F_{i}^{(s)}(1,1,1,...,1),F_{i}^{(s+1)}(\underbrace{1,1,1,...1}_{s~\footnotesize\mbox{~terms}},2), F_{i}^{(r)}(\underbrace{1,1,1,...1}_{s~\footnotesize\mbox{~terms}},2,...,2),x_{i})$$
$$\mathbf{1}\big[(y^{(1)}_{1:n},y^{(2)}_{1:n},F_{1:n}^{(1)}(1), ...,F_{i}^{(s)}(1,1,1,...,1),F_{i}^{(s+1)}(\underbrace{1,1,1,...1}_{s~\footnotesize\mbox{~terms}},2), ...,F_{1:n}^{(r)}(\underbrace{1,1,1,...1}_{s~\footnotesize\mbox{~terms}},2,...,2),x_{1:n})\in \mathcal{T}_{\epsilon}^{(n)}~\big]\bigg).$$
When $s=0$, the covariance is exactly zero. For $s>0$, we can use the inequality $Cov(X,Y)\leq \mathbf{E}(XY)$ for non-negative random variables to bound the covariance term
$$Cov\bigg($$$$
\prod_{i=1}^nq(y^{(1)}_{i},y^{(2)}_{i}|F_{i}^{(1)}(1), ...,F_{i}^{(s)}(1,1,1,...,1),F_{i}^{(s+1)}(1,1,1,...,1), F_{i}^{(r)}(1,1,1,...,1),x_{i})$$$$\mathbf{1}\big[(y^{(1)}_{1:n},y^{(2)}_{1:n},F_{1:n}^{(1)}(1), ...,F_{i}^{(s)}(1,1,1,...,1),F_{i}^{(s+1)}(1,1,1,...,1), ...,F_{1:n}^{(r)}(1,1,1,...,1),x_{1:n})\in \mathcal{T}_{\epsilon}^{(n)}~\big]$$$$~~~\textbf{,}\mbox{ and }~~~$$$$
\prod_{i=1}^nq(y^{(1)}_{i},y^{(2)}_{i}|F_{i}^{(1)}(1), ...,F_{i}^{(s)}(1,1,1,...,1),F_{i}^{(s+1)}(\underbrace{1,1,1,...1}_{s~\footnotesize\mbox{~terms}},2), F_{i}^{(r)}(\underbrace{1,1,1,...1}_{s~\footnotesize\mbox{~terms}},2,...,2),x_{i})$$
$$\mathbf{1}\big[(y^{(1)}_{1:n},y^{(2)}_{1:n},F_{1:n}^{(1)}(1), ...,F_{i}^{(s)}(1,1,1,...,1),F_{i}^{(s+1)}(\underbrace{1,1,1,...1}_{s~\footnotesize\mbox{~terms}},2), ...,F_{1:n}^{(r)}(\underbrace{1,1,1,...1}_{s~\footnotesize\mbox{~terms}},2,...,2),x_{1:n})\in \mathcal{T}_{\epsilon}^{(n)}~\big]\bigg)$$
from above by
$$\mathbf{E}\bigg(
\prod_{i=1}^nq(y^{(1)}_{i},y^{(2)}_{i}|F_{i}^{(1)}(1), ...,F_{i}^{(s)}(1,1,1,...,1),F_{i}^{(s+1)}(1,1,1,...,1), F_{i}^{(r)}(1,1,1,...,1),x_{i})\cdot$$$$\mathbf{1}\big[(y^{(1)}_{1:n},y^{(2)}_{1:n},F_{1:n}^{(1)}(1), ...,F_{i}^{(s)}(1,1,1,...,1),F_{i}^{(s+1)}(1,1,1,...,1), ...,F_{1:n}^{(r)}(1,1,1,...,1),x_{1:n})\in \mathcal{T}_{\epsilon}^{(n)}~\big]\cdot$$$$
\prod_{i=1}^nq(y^{(1)}_{i},y^{(2)}_{i}|F_{i}^{(1)}(1), ...,F_{i}^{(s)}(1,1,1,...,1),F_{i}^{(s+1)}(\underbrace{1,1,1,...1}_{s~\footnotesize\mbox{~terms}},2), F_{i}^{(r)}(\underbrace{1,1,1,...1}_{s~\footnotesize\mbox{~terms}},2,...,2),x_{i})\cdot
$$$$\mathbf{1}\big[(y^{(1)}_{1:n},y^{(2)}_{1:n},F_{1:n}^{(1)}(1), ...,F_{i}^{(s)}(1,1,1,...,1),F_{i}^{(s+1)}(\underbrace{1,1,1,...1}_{s~\footnotesize\mbox{~terms}},2), ...,F_{1:n}^{(r)}(\underbrace{1,1,1,...1}_{s~\footnotesize\mbox{~terms}},2,...,2),x_{1:n})\in \mathcal{T}_{\epsilon}^{(n)}~\big]\bigg).$$
We can bound this further from above by
$$\mathbf{E}\bigg(
\prod_{i=1}^nq(y^{(1)}_{i},y^{(2)}_{i}|F_{i}^{(1)}(1), ...,F_{i}^{(s)}(1,1,1,...,1),F_{i}^{(s+1)}(1,1,1,...,1), F_{i}^{(r)}(1,1,1,...,1),x_{i})\cdot$$$$
\prod_{i=1}^nq(y^{(1)}_{i},y^{(2)}_{i}|F_{i}^{(1)}(1), ...,F_{i}^{(s)}(1,1,1,...,1),F_{i}^{(s+1)}(\underbrace{1,1,1,...1}_{s~\footnotesize\mbox{~terms}},2), F_{i}^{(r)}(\underbrace{1,1,1,...1}_{s~\footnotesize\mbox{~terms}},2,...,2),x_{i})\cdot
$$$$\mathbf{1}\big[(y^{(1)}_{1:n},y^{(2)}_{1:n},F_{1:n}^{(1)}(1), ...,F_{i}^{(s)}(1,1,1,...,1),x_{1:n})\in \mathcal{T}_{\epsilon}^{(n)}~\big]\bigg).$$
Let us compute this expectation by first conditioning it on $F_{i}^{(1)}(1), ...,F_{i}^{(s)}(1,1,1,...,1)$:
$$\mathbf{E}\bigg(
\prod_{i=1}^nq(y^{(1)}_{i},y^{(2)}_{i}|F_{i}^{(1)}(1), ...,F_{i}^{(s)}(1,1,1,...,1),F_{i}^{(s+1)}(1,1,1,...,1), F_{i}^{(r)}(1,1,1,...,1),x_{i})\cdot$$$$
\prod_{i=1}^nq(y^{(1)}_{i},y^{(2)}_{i}|F_{i}^{(1)}(1), ...,F_{i}^{(s)}(1,1,1,...,1),F_{i}^{(s+1)}(\underbrace{1,1,1,...1}_{s~\footnotesize\mbox{~terms}},2), F_{i}^{(r)}(\underbrace{1,1,1,...1}_{s~\footnotesize\mbox{~terms}},2,...,2),x_{i})\cdot$$$$\mathbf{1}\big[(y^{(1)}_{1:n},y^{(2)}_{1:n},F_{1:n}^{(1)}(1), ...,F_{i}^{(s)}(1,1,1,...,1),x_{1:n})\in \mathcal{T}_{\epsilon}^{(n)}~\big]\bigg)=$$
$$\mathbf{E}\bigg(\mathbf{E}\bigg(
\prod_{i=1}^nq(y^{(1)}_{i},y^{(2)}_{i}|F_{i}^{(1)}(1), ...,F_{i}^{(s)}(1,1,1,...,1),F_{i}^{(s+1)}(1,1,1,...,1), F_{i}^{(r)}(1,1,1,...,1),x_{i})\cdot$$$$
\prod_{i=1}^nq(y^{(1)}_{i},y^{(2)}_{i}|F_{i}^{(1)}(1), ...,F_{i}^{(s)}(1,1,1,...,1),F_{i}^{(s+1)}(\underbrace{1,1,1,...1}_{s~\footnotesize\mbox{~terms}},2), F_{i}^{(r)}(\underbrace{1,1,1,...1}_{s~\footnotesize\mbox{~terms}},2,...,2),x_{i})\cdot$$$$\mathbf{1}\big[(y^{(1)}_{1:n},y^{(2)}_{1:n},F_{1:n}^{(1)}(1), ...,F_{i}^{(s)}(1,1,1,...,1),x_{1:n})\in \mathcal{T}_{\epsilon}^{(n)}~\big]\bigg|F_{i}^{(1)}(1), ...,F_{i}^{(s)}(1,1,1,...,1)\bigg)\bigg)=$$
$$\mathbf{E}\bigg(\prod_{i=1}^nq(y^{(1)}_{i},y^{(2)}_{i}|F_{i}^{(1)}(1), ...,F_{i}^{(s)}(1,1,1,...,1),x_{i})^2$$$$\mathbf{1}\big[(y^{(1)}_{1:n},y^{(2)}_{1:n},F_{1:n}^{(1)}(1), ...,F_{i}^{(s)}(1,1,1,...,1),x_{1:n})\in \mathcal{T}_{\epsilon}^{(n)}~\big]\bigg).$$
Using the fact that when $$(y^{(1)}_{1:n},y^{(2)}_{1:n},F_{1:n}^{(1)}(1), ...,F_{1:n}^{(s)}(1,1,1,...,1),x_{1:n})$$ are jointly typical, the term $$\prod_{i=1}^nq(y^{(1)}_{i},y^{(2)}_{i}|F_{i}^{(1)}(1), ...,F_{i}^{(s)}(1,1,1,...,1),x_{i})$$ is less than or equal to
$2^{-n(H(Y^{(1)}_{i},Y^{(2)}_{i}|F^{(1)}, ..., F^{(s)},X)-\epsilon)}$, we can further bound the covariance term from above by
$$2^{-n(H(Y^{(1)}_{i},Y^{(2)}_{i}|F^{(1)}, ..., F^{(s)},X)-\epsilon)}\mathbf{E}\bigg(\prod_{i=1}^nq(y^{(1)}_{i},y^{(2)}_{i}|F_{i}^{(1)}(1), ...,F_{i}^{(s)}(1,1,1,...,1),x_{i})$$$$\mathbf{1}\big[(y^{(1)}_{1:n},y^{(2)}_{1:n},F_{1:n}^{(1)}(1), ...,F_{i}^{(s)}(1,1,1,...,1),x_{1:n})\in \mathcal{T}_{\epsilon}^{(n)}~\big]\bigg)\leq$$
$$2^{-n(H(Y^{(1)}_{i},Y^{(2)}_{i}|F^{(1)}, ..., F^{(s)},X)-\epsilon)}\mathbf{E}\bigg(\prod_{i=1}^nq(y^{(1)}_{i},y^{(2)}_{i}|F_{i}^{(1)}(1), ...,F_{i}^{(s)}(1,1,1,...,1),x_{i})\bigg)=$$
$$2^{-n(H(Y^{(1)}_{i},Y^{(2)}_{i}|F^{(1)}, ..., F^{(s)},X)-\epsilon)}\prod_{i=1}^nq(y^{(1)}_{i},y^{(2)}_{i}|x_i).$$
Since $x_{1:n},y^{(1)}_{1:n}, y^{(2)}_{1:n}$ is jointly typical, $\prod_{i=1}^nq(y^{(1)}_{i},y^{(2)}_{i}|x_i)$ is less than or equal to $2^{-n(H(Y^{(1)}, Y^{(2)}|X)-\epsilon)}$. Therefore the covariance term is less than or equal to $2^{-n(H(Y^{(1)}_{i},Y^{(2)}_{i}|F^{(1)}, ..., F^{(s)},X)-\epsilon)-n(H(Y^{(1)}, Y^{(2)}|X)-\epsilon)}$.

Thus the variance of $\hat{P}_1(x_{1:n}, y^{(1)}_{1:n}, y^{(2)}_{1:n})$ is less than or equal to the summation of $s$ from $0$ and $r$ of
$$\prod_{i=1}^nq(x_i)^22^{-n\sum_{i=1}^s R'_i}2^{-n(H(Y^{(1)}_{i},Y^{(2)}_{i}|F^{(1)}, ..., F^{(s)},X)-\epsilon)-n(H(Y^{(1)}, Y^{(2)}|X)-\epsilon)}.$$
Since $x_{1:n}$ is typical, $\prod_{i=1}^nq(x_i)$ is less than or equal to $2^{-2n(H(X)-\epsilon)}$. Thus,
the variance of $\hat{P}_1(x_{1:n}, y^{(1)}_{1:n}, y^{(2)}_{1:n})$ is less than or equal to the summation of $s$ from $0$ and $r$ of
$$2^{-2n(H(X)-\epsilon)-n\sum_{i=1}^s R'_i-n(H(Y^{(1)}_{i},Y^{(2)}_{i}|F^{(1)}, ..., F^{(s)},X)-\epsilon)-n(H(Y^{(1)}, Y^{(2)}|X)-\epsilon)}.$$
Remember that the total variance was shown to be bounded from above by $$\sum_{(x_{1:n}, y^{(1)}_{1:n}, y^{(2)}_{1:n})\in \mathcal{T}_{\epsilon}^{(n)}}\sqrt{\textbf{Var}\big(\hat{P}(x_{1:n}, y^{(1)}_{1:n}, y^{(2)}_{1:n})\big)}.$$
Our upper bound on the variance of $\hat{P}_1(x_{1:n}, y^{(1)}_{1:n}, y^{(2)}_{1:n})$ does not depend on the particular realization of $x_{1:n}, y^{(1)}_{1:n}, y^{(2)}_{1:n}$. Since the number of jointly typical sequences of $(x_{1:n}, y^{(1)}_{1:n}, y^{(2)}_{1:n})$ is bounded from above by $2^{n(H(X,Y^{(1)}, Y^{(2)})+\epsilon)}$, the total variance will be less than or equal to square root of the summation of $s$ from $0$ and $r$ of
$$2^{2n(H(X,Y^{(1)}, Y^{(2)})+\epsilon)-2n(H(X)-\epsilon)-n\sum_{i=1}^s R'_i-n(H(Y^{(1)}_{i},Y^{(2)}_{i}|F^{(1)}, ..., F^{(s)},X)-\epsilon)-n(H(Y^{(1)}, Y^{(2)}|X)-\epsilon)}=$$
$$2^{-n(\sum_{i=1}^s R'_i-I(Y^{(1)}, Y^{(2)};F^{(1)}, ..., F^{(s)}|X)-7\epsilon)}.$$
This term converges to zero for small values of $\epsilon$ since $\sum_{i=1}^s R'_i>I(Y^{(1)}, Y^{(2)};F^{(1)}, ..., F^{(s)}|X)$.
\end{proof}
\begin{proof}[Proof of Theorem \ref{Thm:Achievability2}] Take some $\epsilon>0$. One can then find a natural number $n_{0}$ such that for any $n>n_0$, there are $(n,\epsilon)$ codes for both the problems of generating $q(z^{(1)},z^{(2)}|x^{(1)},x^{(2)})$ and $q(y^{(1)},y^{(2)}|x^{(1)}z^{(1)},x^{(2)}z^{(2)})$ with the given rounds of communication and rate pair constraints achieving total variation distances less than or equal to $\epsilon$. We claim that concatenating these codes would give us a $(n,2\epsilon)$ code for the problem of generating $q(y^{(1)},y^{(2)}|x^{(1)},x^{(2)})$ with $r_1+r_2$ rounds of communication and the rate pair $(R_{12}+R'_{12}, R_{21}+R'_{21})$. First of all, one can run the concatenated code with  $r_1+r_2$ rounds of communication. If $r_1$ is even, then it will be the first party's turn to start the simulation of the second code. If  $r_1$ is odd, one can interpret the last communication of the first party in the first code and the first communication of the first party in the second code as a single round of communication. This adds up to $r_1+r_2-1$ rounds of communication for the concatenated code which is even better. The rate of the concatenated code is clearly $(R_{12}+R'_{12}, R_{21}+R'_{21})$. So, it remains to check the total variation distance in the concatenated code. We know that
\begin{align*}\sum_{x^{(1)}_{1:n}, x^{(2)}_{1:n}, z^{(1)}_{1:n}, z^{(2)}_{1:n}}\big|&p(X^{(1)}_{1:n}=x^{(1)}_{1:n}, X^{(2)}_{1:n}=x^{(2)}_{1:n}, \widehat{Z}^{(1)}_{1:n}=z^{(1)}_{1:n},
\widehat{Z}^{(2)}_{1:n}=z^{(2)}_{1:n})-\\&p(X^{(1)}_{1:n}=x^{(1)}_{1:n}, X^{(2)}_{1:n}=x^{(2)}_{1:n}, Z^{(1)}_{1:n}=z^{(1)}_{1:n},
Z^{(2)}_{1:n}=z^{(2)}_{1:n})\big|\leq \epsilon.
\end{align*}
\begin{align*}\sum_{x^{(1)}_{1:n}, x^{(2)}_{1:n}, z^{(1)}_{1:n}, z^{(2)}_{1:n}, y^{(1)}_{1:n}, y^{(2)}_{1:n}}\bigg|&\bigg(p(\widehat{Y}^{(1)}_{1:n}=y^{(1)}_{1:n},
\widehat{Y}^{(2)}_{1:n}=y^{(2)}_{1:n}|X^{(1)}_{1:n}=x^{(1)}_{1:n}, X^{(2)}_{1:n}=x^{(2)}_{1:n}, \widehat{Z}^{(1)}_{1:n}=z^{(1)}_{1:n},
\widehat{Z}^{(2)}_{1:n}=z^{(2)}_{1:n})-\\&p(Y^{(1)}_{1:n}=y^{(1)}_{1:n},
Y^{(2)}_{1:n}=y^{(2)}_{1:n}|X^{(1)}_{1:n}=x^{(1)}_{1:n}, X^{(2)}_{1:n}=x^{(2)}_{1:n}, Z^{(1)}_{1:n}=z^{(1)}_{1:n},
Z^{(2)}_{1:n}=z^{(2)}_{1:n})\bigg)\times\\&p(X^{(1)}_{1:n}=x^{(1)}_{1:n}, X^{(2)}_{1:n}=x^{(2)}_{1:n}, Z^{(1)}_{1:n}=z^{(1)}_{1:n},
Z^{(2)}_{1:n}=z^{(2)}_{1:n})\bigg|\leq \epsilon.
\end{align*}
Multiply the first equation by $$p(\widehat{Y}^{(1)}_{1:n}=y^{(1)}_{1:n},
\widehat{Y}^{(2)}_{1:n}=y^{(2)}_{1:n}|X^{(1)}_{1:n}=x^{(1)}_{1:n}, X^{(2)}_{1:n}=x^{(2)}_{1:n}, \widehat{Z}^{(1)}_{1:n}=z^{(1)}_{1:n},
\widehat{Z}^{(2)}_{1:n}=z^{(2)}_{1:n})$$and sum it up over all $y^{(1)}_{1:n}$, and $y^{(2)}_{1:n}$; then add this with the second equation above and use the triangle inequality to conclude that
\begin{align*}\sum_{x^{(1)}_{1:n}, x^{(2)}_{1:n}, z^{(1)}_{1:n}, z^{(2)}_{1:n}, y^{(1)}_{1:n}, y^{(2)}_{1:n}}\big|&p(\widehat{Y}^{(1)}_{1:n}=y^{(1)}_{1:n},
\widehat{Y}^{(2)}_{1:n}=y^{(2)}_{1:n},X^{(1)}_{1:n}=x^{(1)}_{1:n}, X^{(2)}_{1:n}=x^{(2)}_{1:n}, \widehat{Z}^{(1)}_{1:n}=z^{(1)}_{1:n},
\widehat{Z}^{(2)}_{1:n}=z^{(2)}_{1:n})-\\&p(Y^{(1)}_{1:n}=y^{(1)}_{1:n},
Y^{(2)}_{1:n}=y^{(2)}_{1:n},X^{(1)}_{1:n}=x^{(1)}_{1:n}, X^{(2)}_{1:n}=x^{(2)}_{1:n}, Z^{(1)}_{1:n}=z^{(1)}_{1:n},
Z^{(2)}_{1:n}=z^{(2)}_{1:n})\big|\leq \epsilon.
\end{align*}
\end{proof}

%
\newpage
\appendices
\section{}\label{AppendixA}
In this appendix, we show that $\sum_z p(z)d_z\leq \epsilon$. Take
an arbitrary $1\leq j\leq n$,
\begin{eqnarray*}&\epsilon\geq \sum_{x^{(1)}_{1:n}, x^{(2)}_{1:n}, y^{(1)}_{1:n}, x^{(2)}_{1:n}}\bigg|
p(X^{(1)}_{1:n}=x^{(1)}_{1:n}, X^{(2)}_{1:n}=x^{(2)}_{1:n},
\widehat{Y}^{(1)}_{1:n}=y^{(1)}_{1:n},
\widehat{Y}^{(2)}_{1:n}=y^{(2)}_{1:n})-\\&-p(X^{(1)}_{1:n}=x^{(1)}_{1:n},
X^{(2)}_{1:n}=x^{(2)}_{1:n}, Y^{(1)}_{1:n}=y^{(1)}_{1:n},
  Y^{(2)}_{1:n}=y^{(2)}_{1:n})\bigg|=\\&
=\sum_{x^{(1)}_{1:n}, x^{(2)}_{1:n}, y^{(1)}_{1:n},
x^{(2)}_{1:n}}p(X^{(1)}_{1:j-1}=x^{(1)}_{1:j-1},
X^{(2)}_{j+1:n}=x^{(2)}_{j+1:n})\times \\&\times\bigg|
p(X^{(1)}_{j:n}=x^{(1)}_{j:n}, X^{(2)}_{1:j}=x^{(2)}_{1:j},
\widehat{Y}^{(1)}_{1:n}=y^{(1)}_{1:n},
\widehat{Y}^{(2)}_{1:n}=y^{(2)}_{1:n}|X^{(1)}_{1:j-1}=x^{(1)}_{1:j-1},
X^{(2)}_{j+1:n}=x^{(2)}_{j+1:n})-\\&-p(X^{(1)}_{j:n}=x^{(1)}_{j:n},
X^{(2)}_{1:j}=x^{(2)}_{1:j}, Y^{(1)}_{1:n}=y^{(1)}_{1:n},
  Y^{(2)}_{1:n}=y^{(2)}_{1:n}|X^{(1)}_{1:j-1}=x^{(1)}_{1:j-1},
X^{(2)}_{j+1:n}=x^{(2)}_{j+1:n})\bigg|\geq
\\&
\sum_{x^{(1)}_{1:j-1},
x^{(2)}_{j+1:n}}p(X^{(1)}_{1:j-1}=x^{(1)}_{1:j-1},
X^{(2)}_{j+1:n}=x^{(2)}_{j+1:n})\times \sum_{x^{(1)}_{j},
x^{(2)}_{j}, y^{(1)}_{j}, y^{(2)}_{j}}
 \\&\bigg|\sum_{x^{(1)}_{j+1:n},
x^{(2)}_{1:j-1}, y^{(1)}_{1:j-1, j+1:n}, y^{(2)}_{1:j-1,
j+1:n}}\bigg(\\&p(X^{(1)}_{j:n}=x^{(1)}_{j:n},
X^{(2)}_{1:j}=x^{(2)}_{1:j}, \widehat{Y}^{(1)}_{1:n}=y^{(1)}_{1:n},
\widehat{Y}^{(2)}_{1:n}=y^{(2)}_{1:n}|X^{(1)}_{1:j-1}=x^{(1)}_{1:j-1},
X^{(2)}_{j+1:n}=x^{(2)}_{j+1:n})-\\&-p(X^{(1)}_{j:n}=x^{(1)}_{j:n},
X^{(2)}_{1:j}=x^{(2)}_{1:j}, Y^{(1)}_{1:n}=y^{(1)}_{1:n},
  Y^{(2)}_{1:n}=y^{(2)}_{1:n}|X^{(1)}_{1:j-1}=x^{(1)}_{1:j-1},
X^{(2)}_{j+1:n}=x^{(2)}_{j+1:n})\bigg)\bigg|=\\&=
\sum_{x^{(1)}_{1:j-1},
x^{(2)}_{j+1:n}}p(X^{(1)}_{1:j-1}=x^{(1)}_{1:j-1},
X^{(2)}_{j+1:n}=x^{(2)}_{j+1:n})\times \sum_{x^{(1)}_{j},
x^{(2)}_{j}, y^{(1)}_{j}, y^{(2)}_{j}}
 \\&\bigg|p(X^{(1)}_{j}=x^{(1)}_{j},
X^{(2)}_{j}=x^{(2)}_{j}, \widehat{Y}^{(1)}_{j}=y^{(1)}_{j},
\widehat{Y}^{(2)}_{j}=y^{(2)}_{j}|X^{(1)}_{1:j-1}=x^{(1)}_{1:j-1},
X^{(2)}_{j+1:n}=x^{(2)}_{j+1:n})-\\&-p(X^{(1)}_{j}=x^{(1)}_{j},
X^{(2)}_{j}=x^{(2)}_{j}, Y^{(1)}_{j}=y^{(1)}_{j},
  Y^{(2)}_{j}=y^{(2)}_{j}|X^{(1)}_{1:j-1}=x^{(1)}_{1:j-1},
X^{(2)}_{j+1:n}=x^{(2)}_{j+1:n})\bigg)\bigg|
\end{eqnarray*}
Since the above equation holds for any $1\leq j\leq n$, we get
$\sum_z p(z)d_z\leq \epsilon$.

\section{}\label{AppendixB}
In this appendix, we show that the triple $\big(R_{12}, R_{21},
R_{12}+R_{21}-I(Y^{(1)};Y^{(2)}|X^{(1)}X^{(2)})+3\kappa(\epsilon)\big)$
is coordinate by coordinate greater than or equal to
\begin{eqnarray*}&\sum_z p(z)\bigg(I(\overrightarrow{F};X^{(1)}|X^{(2)}, Z=z),
I(\overrightarrow{F};X^{(2)}|X^{(1)}, Z=z),\\&
I(\overrightarrow{F};X^{(1)}|X^{(2)}, Z=z)+
I(\overrightarrow{F};X^{(2)}|X^{(1)}, Z=z)\bigg)\end{eqnarray*}
where
$\kappa(\epsilon)=-\epsilon\log(\frac{\epsilon}{|\mathcal{X}^{(1)}||\mathcal{X}^{(2)}||\mathcal{\widehat{Y}}^{(1)}||\mathcal{\widehat{Y}}^{(2)}|})$
and $I(Y^{(1)};Y^{(2)}|X^{(1)}X^{(2)})$ is evaluated assuming that
$\\X^{(1)}, X^{(2)}, Y^{(1)}, Y^{(2)}$ have the joint distribution
of $q(x^{(1)}, x^{(2)}, y^{(1)}, y^{(2)})$.

\begin{itemize}
  \item \begin{eqnarray*}&R_{12} \geq \frac{1}{n}H(C_1, C_3, C_5, ...)\geq \frac{1}{n}H(C_1, C_3, C_5,
  ...|M_2X^{(2)}_{1:n})=
  \frac{1}{n}H(C_1, C_2, C_3, ...|M_2X^{(2)}_{1:n})\\&\geq
  \frac{1}{n}I(\overrightarrow{C};X^{(1)}_{1:n}|M_2X^{(2)}_{1:n})=
  \frac{1}{n}\sum_{j=1}^nI(\overrightarrow{C};X^{(1)}_{j}|M_2X^{(2)}_{1:n}X^{(1)}_{j+1:n})
  =I(\overrightarrow{F};X^{(1)}|M_2X^{(2)}_{J+1:n}X^{(2)}Z)
  \\&=H(X^{(1)}|M_2X^{(2)}_{J+1:n}X^{(2)}Z)-H(X^{(1)}|M_2X^{(2)}_{J+1:n}X^{(2)}Z\overrightarrow{F})\\&=
  H(X^{(1)}|X^{(2)}Z)-H(X^{(1)}|M_2X^{(2)}_{J+1:n}X^{(2)}Z\overrightarrow{F})=I(M_2X^{(2)}_{J+1:n}\overrightarrow{F};X^{(1)}|X^{(2)}Z)
  \\&\geq I(\overrightarrow{F};X^{(1)}|X^{(2)}Z)\end{eqnarray*}
  \item Similarly we have:\begin{eqnarray*}&R_{21} \geq I(\overrightarrow{F};X^{(2)}|X^{(1)}Z).\end{eqnarray*}
  \item Note that \begin{eqnarray*}&R_{12}+R_{21}\geq
  \frac{1}{n}H(\overrightarrow{C}|M_2X^{(2)}_{1:n})+\frac{1}{n}H(\overrightarrow{C}|M_1X^{(1)}_{1:n})\geq
  \\&
  \frac{1}{n}I(\overrightarrow{C};X^{(1)}_{1:n}|M_2X^{(2)}_{1:n})+ \frac{1}{n}H(\overrightarrow{C}|M_2X^{(2)}_{1:n}X^{(1)}_{1:n})
  +
  \frac{1}{n}I(\overrightarrow{C};X^{(2)}_{1:n}|M_1X^{(1)}_{1:n})+ \frac{1}{n}H(\overrightarrow{C}|M_1X^{(1)}_{1:n}X^{(2)}_{1:n})
  \end{eqnarray*}
  The first and the third term are respectively greater than or equal to
  $I(\overrightarrow{F};X^{(1)}|X^{(2)}Z)$ and
  $I(\overrightarrow{F};X^{(2)}|X^{(1)}Z)$ as shown above. For the second and the fourth
  term, we have:
\begin{eqnarray*}&
\frac{1}{n}H(\overrightarrow{C}|M_2X^{(2)}_{1:n}X^{(1)}_{1:n})+
\frac{1}{n}H(\overrightarrow{C}|M_1X^{(1)}_{1:n}X^{(2)}_{1:n})\geq\\&
I(M_1\overrightarrow{C};M_2\overrightarrow{C}|X^{(2)}_{1:n}X^{(1)}_{1:n})
\end{eqnarray*}
since for any odd $i$,
\begin{eqnarray*}&
I(M_1C_{1:i};M_2C_{1:i}|X^{(2)}_{1:n}X^{(1)}_{1:n})=
I(M_1C_{1:i};M_2C_{1:i-1}|X^{(2)}_{1:n}X^{(1)}_{1:n})+
I(M_1C_{1:i};C_i|X^{(2)}_{1:n}X^{(1)}_{1:n}M_2C_{1:i-1})\\&
=I(M_1C_{1:i-1};M_2C_{1:i-1}|X^{(2)}_{1:n}X^{(1)}_{1:n})+
I(M_1C_{1:i};C_i|X^{(2)}_{1:n}X^{(1)}_{1:n}M_2C_{1:i-1})\\& \leq
I(M_1C_{1:i-1};M_2C_{1:i-1}|X^{(2)}_{1:n}X^{(1)}_{1:n})+
H(C_i|X^{(2)}_{1:n}X^{(1)}_{1:n}M_2C_{1:i-1})
\\& =
I(M_1C_{1:i-1};M_2C_{1:i-1}|X^{(2)}_{1:n}X^{(1)}_{1:n})+
H(C_i|X^{(2)}_{1:n}X^{(1)}_{1:n}M_2C_{1:i-1})+
H(C_i|X^{(2)}_{1:n}X^{(1)}_{1:n}M_1C_{1:i-1}).
\end{eqnarray*}
Similarly, for any even $i$,
\begin{eqnarray*}&
I(M_1C_{1:i};M_2C_{1:i}|X^{(2)}_{1:n}X^{(1)}_{1:n})
\\& \leq
I(M_1C_{1:i-1};M_2C_{1:i-1}|X^{(2)}_{1:n}X^{(1)}_{1:n})+
H(C_i|X^{(2)}_{1:n}X^{(1)}_{1:n}M_2C_{1:i-1})+
H(C_i|X^{(2)}_{1:n}X^{(1)}_{1:n}M_1C_{1:i-1}).
\end{eqnarray*}
Further, $I(M_1;M_2|X^{(2)}_{1:n}X^{(1)}_{1:n})=0$. Hence
\begin{eqnarray*}&
\frac{1}{n}H(\overrightarrow{C}|M_2X^{(2)}_{1:n}X^{(1)}_{1:n})+
\frac{1}{n}H(\overrightarrow{C}|M_1X^{(1)}_{1:n}X^{(2)}_{1:n})\geq\\&
\frac{1}{n}I(M_1\overrightarrow{C};M_2\overrightarrow{C}|X^{(2)}_{1:n}X^{(1)}_{1:n})
\end{eqnarray*}
Since $\widehat{Y}^{(1)}_{1:n}$ is created from
$M_1\overrightarrow{C}X^{(1)}_{1:n}$ and $\widehat{Y}^{(2)}_{1:n}$
is created from $M_2\overrightarrow{C}X^{(2)}_{1:n}$, we have:
\begin{eqnarray*}&
\frac{1}{n}H(\overrightarrow{C}|M_2X^{(2)}_{1:n}X^{(1)}_{1:n})+
\frac{1}{n}H(\overrightarrow{C}|M_1X^{(1)}_{1:n}X^{(2)}_{1:n})\geq\\&
\frac{1}{n}I(\widehat{Y}^{(1)}_{1:n};\widehat{Y}^{(2)}_{1:n}|X^{(2)}_{1:n}X^{(1)}_{1:n}).\end{eqnarray*}
Note that \begin{eqnarray*}&\frac{1}{n}I(\widehat{Y}^{(1)}_{1:n};\widehat{Y}^{(2)}_{1:n}|X^{(2)}_{1:n}X^{(1)}_{1:n})=\\&
\frac{1}{n}H(\widehat{Y}^{(1)}_{1:n}X^{(2)}_{1:n}X^{(1)}_{1:n})+\frac{1}{n}H(\widehat{Y}^{(2)}_{1:n}X^{(2)}_{1:n}X^{(1)}_{1:n})\\&
-\frac{1}{n}H(\widehat{Y}^{(1)}_{1:n}\widehat{Y}^{(2)}_{1:n}X^{(2)}_{1:n}X^{(1)}_{1:n})-\frac{1}{n}H(X^{(2)}_{1:n}X^{(1)}_{1:n}).\end{eqnarray*}
Using the Csisz\'{a}r-K\"{o}rner inequality \cite[Lemma
2.7]{CsiszárKörner}, we can bound this expression from below by
\begin{eqnarray*}&
\frac{1}{n}H(Y^{(1)}_{1:n}X^{(2)}_{1:n}X^{(1)}_{1:n})+\frac{1}{n}H(Y^{(2)}_{1:n}X^{(2)}_{1:n}X^{(1)}_{1:n})\\&
-\frac{1}{n}H(Y^{(1)}_{1:n}Y^{(2)}_{1:n}X^{(2)}_{1:n}X^{(1)}_{1:n})-\frac{1}{n}H(X^{(2)}_{1:n}X^{(1)}_{1:n})+3\frac{\epsilon}{n}\log(\frac{\epsilon}{(|\mathcal{X}^{(1)}||\mathcal{X}^{(2)}||\mathcal{\widehat{Y}}^{(1)}||\mathcal{\widehat{Y}}^{(2)}|)^n})
=\\&
\frac{1}{n}I(Y^{(1)}_{1:n};Y^{(2)}_{1:n}|X^{(2)}_{1:n}X^{(1)}_{1:n})+3\frac{\epsilon}{n}\log(\frac{\epsilon}{(|\mathcal{X}^{(1)}||\mathcal{X}^{(2)}||\mathcal{\widehat{Y}}^{(1)}||\mathcal{\widehat{Y}}^{(2)}|)^n})\geq\\&
\frac{1}{n}I(Y^{(1)}_{1:n};Y^{(2)}_{1:n}|X^{(2)}_{1:n}X^{(1)}_{1:n})-3\kappa(\epsilon)=
I(Y^{(1)};Y^{(2)}|X^{(2)}X^{(1)})-3\kappa(\epsilon).
\end{eqnarray*}
Therefore
\begin{eqnarray*}R_{12}+R_{21}\geq
I(\overrightarrow{F};X^{(1)}|X^{(2)}Z)+I(\overrightarrow{F};X^{(2)}|X^{(1)}Z)+
I(Y^{(1)};Y^{(2)}|X^{(2)}X^{(1)})-3\kappa(\epsilon)
\end{eqnarray*}
\end{itemize}

\section{}\label{AppendixC}
We begin by imposing a cardinality bounds on $|\mathcal{F}_1|$, then
on $|\mathcal{F}_2|$, and so on. Having imposed cardinality
constraints on $|\mathcal{F}_{k-1}|$, $|\mathcal{F}_{k-2}|$, ... and
$|\mathcal{F}_{1}|$ will we show that the cardinality of
$|\mathcal{F}_k|$ can be bounded from above by
$|\mathcal{F}_{k-1}||\mathcal{F}_{k-2}|\cdot\cdot\cdot
|\mathcal{F}_{1}||\mathcal{X}^{(1)}||\mathcal{Y}^{(1)}||\mathcal{X}^{(2)}||\mathcal{Y}^{(2)}|+1$.

In order to impose the cardinality bound on $F_1$, we fix $p(f_2,
..., f_r, x^{(1)}, x^{(2)}, \widehat{y}^{(1)},
\widehat{y}^{(2)}|f_1)$ and vary the marginal distribution $p(f_1)$.
The equations $I(F_i;X^{(2)}|F_{1:i-1}X^{(1)})=0$ for odd $i$,
$I(F_i;X^{(1)}|F_{1:i-1}X^{(2)})=0$ for even $i$,
$I(\widehat{Y}^{(1)};X^{(2)}\widehat{Y}^{(2)}|\overrightarrow{F}X^{(1)})=0$
and $I(\widehat{Y}^{(2)};X^{(1)}|\overrightarrow{F}X^{(2)})=0$ hold
irrespective of $p(f_1)$. We need to impose
$|\mathcal{X}^{(1)}||\mathcal{\widehat{Y}}^{(1)}||\mathcal{X}^{(2)}||\mathcal{\widehat{Y}}^{(2)}|-1$
equations ensure that the joint distribution of $p(x^{(1)}, x^{(2)},
\widehat{y}^{(1)}, \widehat{y}^{(2)})$ does not change, one equation
for the probability constraint $\sum_{f_1}p(f_1)=1$, and one
equation corresponding to the two terms $I(X^{(1)};
\overrightarrow{F}|X^{(2)})$ and $I(X^{(2)}; \overrightarrow{F}|
X^{(1)})$ (since we are using the generalized Carath\'{e}odory
theorem of Fenchel). Therefore we get the cardinality
bound of
$|\mathcal{X}^{(1)}||\mathcal{\widehat{Y}}^{(1)}||\mathcal{X}^{(2)}||\mathcal{\widehat{Y}}^{(2)}|+1$
on $|\mathcal{F}_{1}|$. In order to find cardinality bounds for
$|\mathcal{F}_2|$, we fix $$p(f_1, f_3, ..., f_r, x^{(1)}, x^{(2)},
\widehat{y}^{(1)}, \widehat{y}^{(2)}|f_2)$$ and vary the marginal
distribution $p(f_2)$. The equations
$I(F_i;X^{(2)}|F_{1:i-1}X^{(1)})=0$ for odd $i\geq 3$,
$I(F_i;X^{(1)}|F_{1:i-1}X^{(2)})=0$ for even $i\geq 3$,
$I(\widehat{Y}^{(1)};X^{(2)}\widehat{Y}^{(2)}|\overrightarrow{F}X^{(1)})=0$
and $I(\widehat{Y}^{(2)};X^{(1)}|\overrightarrow{F}X^{(2)})=0$ hold
irrespective of $p(f_2)$. Next, we impose
$|\mathcal{F}_{1}||\mathcal{X}^{(1)}||\mathcal{\widehat{Y}}^{(1)}||\mathcal{X}^{(2)}||\mathcal{\widehat{Y}}^{(2)}|-1$
equations ensure that the joint distribution of $p(f_1, x^{(1)},
x^{(2)}, \widehat{y}^{(1)}, \widehat{y}^{(2)})$ does not change.
This implies that the equations $I(F_i;X^{(2)}|F_{1:i-1}X^{(1)})=0$
for odd $i< 3$, $I(F_i;X^{(1)}|F_{1:i-1}X^{(2)})=0$ for even $i< 3$
hold. This argument can be repeated for $F_3$, $F_4$, ... and so on.

\section*{Acknowledgment}
The research was partially supported by the NSF grants
CCF-0500234, CCF-0635372, CNS-0627161, and
CNS-0910702, the ARO MURI grant
W911NF-08-1-0233
``Tools for the Analysis and Design of Complex Multi-Scale Networks",
and the NSF Science and Technology Center grant CCF-0939370,
``Science of Information".

\end{document}